\documentclass[11pt,a4paper]{article} 
\usepackage{graphicx}
\usepackage[margin=1in]{geometry} 
\usepackage{authblk} 
\usepackage[square,numbers]{natbib} 
\usepackage{amsthm}
\usepackage{booktabs} 
\usepackage[linesnumbered,vlined,commentsnumbered,ruled]{algorithm2e} 

\SetAlFnt{\small}
\SetAlCapFnt{\small}
\SetAlCapNameFnt{\small}
\SetAlCapHSkip{0pt}
\IncMargin{-\parindent}

\setcitestyle{numbers,square}

\usepackage{amsfonts}
\usepackage{xcolor}
\usepackage{enumitem} 
\usepackage{multirow}
\usepackage{bbm}
\usepackage{amsmath}

\usepackage[switch]{lineno}

\newtheorem{theorem}{Theorem}

\newtheorem{claim}{Claim}
\newtheorem{example}{Example}
\newtheorem{definition}{Definition}
\newtheorem{observation}{Observation}
 \newtheorem{lemma}{Lemma}

\newcommand{\MMS}{\mathsf{MMS}}

\newcommand{\Div}{\mathsf{div}}

\newcommand{\cM}{\mathcal{M}}

\newcommand{\len}{{\rm len}}

\newcommand{\fv}{\mathbf{v}}

\newcommand{\A}{\mathbf{A}}
\newcommand{\B}{\mathbf{B}}

\begin{document}

\title{Approximate Maximin Share with Subjective Divisibility: Beating the 1/2 Barrier}

\author[1]{Xiaohui Bei}
\author[2]{Ke Ding}
\author[2]{Bo Li}
\author[2]{Fangxiao Wang}

\affil[1]{Nanyang Technological University \\ \texttt{xhbei@ntu.edu.sg}}
\affil[2]{The Hong Kong Polytechnic University \\ \texttt{\{cocoding77, comp-bo.li, fangxiao.wang\}@polyu.edu.hk}}

\date{} 

\maketitle

\begin{abstract}

Maximin share (MMS) stands out as a central notion in fair resource allocation. It is known that exact MMS fairness is not always attainable, especially when agents differ along two dimensions: their valuations and their perceptions of the divisibility of resources. The former case with heterogeneous valuations has been widely studied in the literature. The latter, referred to as \emph{subjective divisibility} by Bei et al., [Games Econ. Behav. 2025], remains much less explored.

We study MMS approximation under subjective divisibility. First, we prove that even in the unary valuation setting, where all items have equal value, the optimal approximation ratio is $\frac{2}{3}$. This result is somewhat surprising since in the objective setting, even when agents have heterogeneous valuations, the best possible approximation ratio is at least $\frac{7}{9}$ [Huang and Zhou, 2025]. We then address the general case with both valuation heterogeneity and subjective divisibility. Previous work shows the existence of a $\frac{1}{2}$-approximate MMS allocation. In this paper, we develop new algorithmic techniques that overcome the difficulties posed by subjective divisibility, and improve the approximation guarantee to $\frac{5}{9}$.
Finally, we complement this result with small-agent cases. For up to four agents, we give polynomial-time algorithms that compute $\frac{2}{3}$-approximate MMS fair allocations. These bounds are tight. 

Our results deepen the understanding of MMS fairness under heterogeneous valuations and subjective divisibility, and provide a new perspective for this emerging model.
\end{abstract}

\begin{titlepage}

\maketitle


\setcounter{tocdepth}{2} 
\tableofcontents

\end{titlepage}

\section{Introduction}

Fair division studies how to allocate resources among multiple agents with different valuations in a fair way.
Classic work focused on divisible resource allocation, exemplified by the cake-cutting problem, where fairness notions like proportionality (PROP)~\citep{steinhaus1949division} can always be satisfied and efficiently computed\footnote{
PROP requires that each agent receives a portion of resources that is valued by no less than $\frac{1}{n}$ fraction of her value for all resources, where $n$ is the number of agents.}.
In contrast, when goods are indivisible, PROP may become unattainable. This has motivated relaxations of the concept such as the \textit{maximin share guarantee} (MMS)~\citep{DBLP:conf/bqgt/Budish10}, which requires each agent receives at least the value she could secure by partitioning the goods herself and picking last.
Although exact MMS allocations cannot always be guaranteed~\citep{DBLP:journals/jacm/KurokawaPW18,DBLP:conf/wine/FeigeST21,DBLP:conf/soda/AkramiG24,DBLP:journals/corr/abs-2510-10423}, constant approximations are known, with the best current bound $\frac{7}{9}$ \citep{DBLP:journals/corr/abs-2511-13056}. 

While most existing studies assume goods are either fully divisible or indivisible, many real-world scenarios require allocating both types of resources at the same time.
For example, when dividing marital or heritage assets, there may be divisible goods such as money and indivisible ones like a house or a car. To formalize this, \citet{DBLP:journals/ai/BeiLLLL21} introduced the mixed goods model, which considers allocations involving both divisible and indivisible items. This model has since inspired a line of work on fairness in the mixed setting \citep{DBLP:journals/aamas/BeiLLW21,DBLP:conf/approx/BhaskarSV21,DBLP:conf/ijcai/LiLLT23,DBLP:journals/corr/abs-2410-06877,DBLP:conf/ijcai/Li0LW24}.

While the mixed goods model captures both divisible and indivisible resources, it assumes agents agree on which items are divisible. In practice, this is often not the case. A given resource may be viewed as divisible by one agent but indivisible by another, depending on how it is used. For example, in the canonical cake-cutting model, a cake may be divisible for someone who wants to consume it, but indivisible for someone who wishes to gift it as a whole. Similarly, a time slot for using a university venue may be indivisible to a professor teaching a lecture, but divisible to a student doing her course assignment; a bundle of building materials may be divisible to a reseller but indivisible to a builder who needs the entire quantity to complete a project. These cases reflect a more general phenomenon called \textit{subjective divisibility}, where the perceived divisibility of a resource varies across agents.
\citet{DBLP:journals/corr/abs-2310-00976} initiated the study of fair division under subjective divisibility. This setting generalizes the mixed goods model by allowing each agent to have her own view on whether a good is divisible. It introduces a new dimension of heterogeneity into the problem: beyond agents' differing valuations, there is now disagreement on the structure of the resources themselves.

This richer model raises new theoretical questions. When agents are identical both in valuation and divisibility view, MMS fairness coincides with egalitarian welfare and is always attainable. When heterogeneity arises only from valuations, the problem is well studied under the divisible, indivisible, and mixed settings. However, when agents differ in how they perceive the divisibility of items, much less is known. To our knowledge, the only prior result is due to \citet{DBLP:journals/corr/abs-2310-00976}, who showed that MMS fairness becomes significantly harder: even with identical values, exact MMS may not be achievable. They proposed a $\frac12$-approximate MMS allocation for general instances and a $\frac23$-approximation for three agents.

In this paper, we study how subjective divisibility -- alone and in combination with valuation heterogeneity -- affects the feasibility and approximation of MMS fairness. We present new techniques to design algorithms that improve the existing approximation bounds. Our results are summarized in the next subsection.

\subsection{Contribution Overview}

We study the problem of MMS fair allocation with subjective divisibility, where agents may have different opinions on whether an item is divisible. 
To understand the impact of subjective divisibility, we first examine a simple setting in which all items have equal value (i.e., unary valuations).
Surprisingly, even in this basic case, the optimal approximation ratio of MMS fairness is only $\frac{2}{3}$ (see Theorem \ref{the:una}).
The proof uses a matching-based argument and a careful classification of agents by the number of items they view as divisible.
This result stands in contrast to the classic indivisible setting with heterogeneous valuations,
where $\frac{7}{9}$-MMS is always achievable \citep{DBLP:journals/corr/abs-2511-13056}.

We then study the general setting with heterogeneous valuations and subjective divisibility.
In this setting, \citet{DBLP:journals/corr/abs-2310-00976} gave an algorithm that guarantees a $\frac{1}{2}$-approximate MMS allocation.
Improving this guarantee is challenging, as many standard techniques often fail. For example, in the indivisible setting, assigning a single item to an agent does not affect the MMS of others. 
Under subjective divisibility, however, removing an item can substantially decrease an agent's MMS value in the reduced instance. 
As a result, the commonly used large item reduction technique for indivisible goods becomes inapplicable.
Similarly, the identical ordering (IDO) reduction, which aligns all agents' value orderings and is widely used for indivisible items \citep{DBLP:journals/teco/BarmanK20}, cannot preserve agents' divisibility views and thus cannot be used in our setting.

To address these challenges, we develop new techniques that control how each allocation affects the remaining agents. 
We begin by identifying large and shareable items and allocating them to a carefully chosen subset of agents.
Agents who receive large items are easily satisfied, while those who receive shareable items -- denoted by $X$ -- will participate in the subsequent phases of the algorithm.
For the remaining agents, intuitively, the problem could be more manageable if few items are divisible (since we can leverage methods from the indivisible setting) or if few items have large values (since it is easier to balance agents' values using small items). 
Based on this insight, our algorithm characterizes a set of {\em critical} agents, denoted by $Z$, who have a number of medium-value items, with a sufficient proportion of these items being divisible. 
The remaining agents are denoted by $Y$.
For the critical agents in $Z$, we prove that it is always possible to allocate each of them two medium-value items without adversely affecting agents in $X$ and $Y$ -- a key technical step in our approach. 
For agents in $X$ and $Y$, we employ a two-step algorithm that treats all remaining items as indivisible, allowing us to apply techniques such as IDO reduction to the resulting sub-instance and thereby simplifying the analysis.
Another technical challenge lies in establishing bounds on the guaranteed value for agents in $Y$, for which we adopt different approaches to bound their values of the items allocated to the other agents.  
This layered approach enables us to navigate the complexities introduced by subjective divisibility.

Beyond the general case of an arbitrary number of agents, we also study instances with only three or four agents (See Theorems \ref{3.1} and \ref{4agent}).
For three agents, \citet{DBLP:journals/corr/abs-2310-00976} established the existence of a tight $\frac{2}{3}$-approximate MMS allocation.
We provide an alternative proof for this result that also yields a polynomial-time algorithm. 
Extending this approach to four agents requires significantly more technical care, but we show that the same $\frac{2}{3}$ approximation remains tight for four agents.
To achieve this, we introduce the notions of nice and temporary reductions, along with a check-and-withdraw strategy for revising temporary reductions. These tools enable us to characterize the optimal structures of MMS allocations across various scenarios.
Determining the optimal approximation ratio for five or more agents remains an intriguing open problem.

\subsection{Related Work}
\label{sec:intro:related}
Since the seminal work of \citet{steinhaus1949division}, the theory of fair allocation has been extensively studied. Early research focused on divisible resources, which lead to the well-known cake-cutting problem \citep{books/daglib/0017734,Procaccia_cake_16,LindnerR16}. In contrast, allocating indivisible items poses greater challenges, as ideal fairness notions such as envy-freeness \citep{GS58,Varian74} and proportionality \citep{steinhaus1949division} are often unachievable. This has motivated a large body of work on relaxed fairness concepts, including envy-freeness up to one or any item (EF1, EFX) \citep{lipton2004approximately,DBLP:conf/bqgt/Budish10,DBLP:journals/teco/CaragiannisKMPS19}, and maximin share (MMS) fairness \citep{DBLP:conf/bqgt/Budish10,DBLP:conf/soda/AkramiG24}, which are more robust in the presence of indivisibilities. 

In the context of MMS fairness, \citet{DBLP:journals/jacm/KurokawaPW18} gave the first $\frac{2}{3}$-approximation algorithm, which was later improved to $\frac{3}{4}$ by \citet{DBLP:conf/sigecom/GhodsiHSSY18} and \citet{DBLP:journals/ai/GargT21}, and more recently to $\frac{10}{13}$ by \citet{DBLP:journals/corr/abs-2510-10423} and $\frac{7}{9}$ by \citet{DBLP:journals/corr/abs-2511-13056}. Beyond the general case, \citet{DBLP:journals/corr/abs-2205-05363} proved the existence of $\frac{11}{12}$-approximate MMS allocations for three agents, and \citet{DBLP:conf/sigecom/GhodsiHSSY18,DBLP:conf/sigecom/BabaioffF22} showed that a $\frac{4}{5}$-approximation exists for four agents. On the negative side, \citet{DBLP:journals/jacm/KurokawaPW18} provided the first counterexample showing that exact MMS allocations may not exist. \citet{DBLP:conf/wine/FeigeST21} strengthened this by proving that no algorithm can guarantee better than a $\frac{39}{40}$-approximation for three agents, $\frac{67}{68}$ for four agents, and $1 - \frac{1}{n^4}$ for general $n$.

Mixed divisibility has emerged as an important direction in fair division and has attracted growing attention in recent years \citep{DBLP:conf/approx/BhaskarSV21,DBLP:journals/corr/abs-2410-06877,DBLP:journals/jair/LiuLSW24}. For envy-freeness, \citet{DBLP:journals/ai/BeiLLLL21} proved the existence of envy-freeness for mixed goods (EFM) in special cases, and approximate-EFM allocations for general instances. \citet{DBLP:journals/corr/abs-2302-13342} showed that a maximum Nash welfare (MNW) allocation satisfies EF1M for agents with additive valuations, and EFXM for those with binary valuations. \citet{DBLP:conf/ijcai/LiLLT23} further designed truthful and EFM mechanisms for several special cases. For maximin share fairness, \citet{DBLP:journals/aamas/BeiLLW21} studied how the presence of divisible goods affects the worst-case MMS approximation. We refer readers to \citet{DBLP:journals/jair/LiuLSW24} for a more comprehensive overview of fair allocation with mixed items.

\section{Preliminaries} \label{sec:prelim}

Let $[k]= \{1,\ldots,k\}$ for any integer $k\ge 1$.
A fair allocation instance consists of $n$ agents $N=\{a_1,\ldots,a_n\}$ and $m$ goods $M=\{g_1, \ldots, g_m\}$.
Each agent $a_i \in N$ has a non-negative utility $v_i(g)$ for each good $g \in M$, and each good is valued positively by at least one agent.

Goods may be divisible for some agents but not for others. For a subset $S\subseteq M$ of items and an agent $a_i\in N$, denote by $S^\text{IND}_i$ and $S^\text{DIV}_i$ the sets of indivisible and divisible items in $S$ for agent $a_i$, respectively. Thus, $M = M^\text{IND}_i \cup M^\text{DIV}_i$ for each $a_i$.

If a good $g$ is indivisible for $a_i$, she gains value $v_i(g)$ only if she receives the entire item; otherwise, she gets zero value. If $g$ is divisible for $a_i$, then her value is proportional to the fraction she receives, i.e., the divisible items are \emph{homogeneous}. To formalize this, each good $g \in M$ is modeled as an interval of length $1$. If agent $a_i$ receives a subinterval $\Tilde{g} \subseteq g$ of length $x \in [0,1]$, denoted by $\Tilde{g} = x \cdot g$, then her utility is
\[
v_i(\tilde{g})= v_i(x\cdot g) = \left\{ 
\begin{array}{ll}
     x \cdot v_i(g) & \textrm{ if $g \in M^\text{DIV}_i$}\\
     v_i(g) & \textrm{ if $g \in M^\text{IND}_i$ and $x = 1$}\\
     0 & \textrm{ if $g \in M^\text{IND}_i$ and $x < 1$}.
\end{array}
\right.
\] 

By $S\subseteq M$, we allow $S$ to contain fractional items. For any $\Tilde{g} \subseteq g$, we write $\Tilde{g} \in S$ to indicate that a fraction $\len(\Tilde{g})$ of item $g$ is included in $S$, and $S\setminus {\Tilde{g}}$ to denote that this fraction is removed from $S$.

For any item $g$ and any $0 \le \alpha \le v_i(g)$, define $\len_i(g,\alpha)$ to be the smallest fraction of $g$ that has value at least $\alpha$ for agent $a_i$:
\[
\len_i(g,\alpha) = \inf\{0 \le x \le \len(g) \mid v_i(x\cdot g) \ge \alpha\}.
\]

We assume that agents have additive utilities: for any $S\subseteq M$, $v_i(S) = \sum_{g \in S} v_i(g)$. Let $\fv=(v_1,\ldots,v_n)$, and denote an instance by $I = \langle N, M, \fv \rangle$.

For any integer $k > 1$, let $\Pi_k(S)$ denote the set of all $k$-partitions of $S$, i.e., tuples $(S_1,\ldots,S_k)$ such that $S_i \cap S_j = \emptyset$ for $i \ne j$ and $\bigcup_{i=1}^k S_i = S$. We allow each $S_i$ to contain fractional items. An \textit{allocation} $\A=(A_1,\ldots,A_n)$ is an $n$-partition of $M$, where $A_i$ is the bundle allocated to agent $a_i$.


We consider the fairness notion of \textit{maximin share (MMS)}. For an instance with subjective divisibility, the MMS of agent $a_i$ is defined as 
\[
\MMS_i=\mathop{\max}_{\A\in \Pi_n(M)}\mathop{\min}_{j \in [n]}v_i(A_j).
\]
A partition $\A = (S_1,\ldots, S_n) \in \Pi_n(M)$ is an \textit{MMS partition} for agent $a_i$ if $\min_{j\in [n]} v_i(S_j) = \MMS_i$.

\begin{definition}[$\alpha$-MMS]
Given $\alpha\in[0,1]$, an allocation $\A= (A_1, \ldots, A_n)$ is called \emph{$\alpha$-approximate MMS fair} (or \emph{$\alpha$-MMS}) if $v_i(A_i) \geq \alpha \cdot \MMS_i$ for every agent $a_i \in N$.
\end{definition}

The following example illustrates the notion of approximate MMS fairness under subjective divisibility. It is adapted from \citet{DBLP:journals/corr/abs-2310-00976}, who used it to establish an upper bound on the achievable approximation ratio.

\begin{example}\label{example:upper}[\citet{DBLP:journals/corr/abs-2310-00976}]
Suppose $N=\{a_1, a_2\}$ and $M=\{g_1, g_2, g_3\}$. The valuations and divisibility of the items are given in Table~\ref{tab:upper}. For agent $a_1$, her MMS partition evenly divides $g_3$ and combines each half with one of $g_1$ or $g_2$: $X_1 = \{g_1, \frac{1}{2}\cdot g_3\}$ and $X_2 = \{g_2, \frac{1}{2}\cdot g_3\}$, which gives $\MMS_1 = 1.5$. Similarly, $\MMS_2 = 1.5$ for $a_2$, where $g_2$ is divided.

However, due to subjectivity, if $g_2$ is divided, agent $a_1$ values it at $0$, and similarly if $g_3$ is divided, agent $a_2$ values it at $0$. Hence, any allocation of the three items must leave at least one agent with a bundle valued at most $1$. This gives an MMS approximation upper bound of $\frac{1}{1.5}=\frac23$.
\end{example}

\begin{table}[h]
\centering
\renewcommand{\arraystretch}{1.2}
\begin{tabular}{c|ccc}
\hline
Agent & Good $g_1$ & Good $g_2$ & Good $g_3$ \\
\hline
$a_1$ & $1$ (IND) & $1$ (IND) & $1$ (DIV) \\
$a_2$ & $1$ (IND) & $1$ (DIV) & $1$ (IND) \\
\hline
\end{tabular}
\medskip
\caption{Valuations and divisibility in Example~\ref{example:upper}. IND and DIV denote indivisible and divisible.}
\label{tab:upper}
\end{table}

\citet{DBLP:journals/corr/abs-2310-00976} also generalized this example to any number of agents and showed that no algorithm can guarantee better than a $\frac{2}{3}$-MMS allocation for any $n \ge 2$.

\section{Warm-Up: Subjective Divisibility with Unary Valuations}
\label{sec:unary}

To understand the impact of subjective divisibility in isolation, we begin with a simplified setting where all agents have identical (unary) valuations. Specifically, we assume $v_i(g_j)=1$ for every agent $a_i\in N$ and item $g_j\in M$. As shown in Example~\ref{example:upper}, no algorithm can guarantee better than a $\frac{2}{3}$-approximate MMS allocation in this setting.

We show that this bound is tight: a $\frac{2}{3}$-MMS allocation always exists and can be computed efficiently.

\begin{theorem}\label{the:una}
The optimal approximation ratio of MMS fairness with unary valuations is $\frac{2}{3}$.
\end{theorem}

\begin{proof}
Let $a=\lfloor m/n \rfloor$ and $b=m \bmod n$, so that $m=a\cdot n + b$ with $0\le b<n$. Let $\Div_i=|M_i^\text{DIV}|$ denote the number of divisible items in $M$ for agent $a_i$. Then the MMS value of $a_i$ is:

\[
\MMS_i =
\begin{cases}
a + \frac{\Div_i}{\Div_i + n - b} & \text{if } \Div_i \le b, \\
a + \frac{b}{n} & \text{if } \Div_i > b.
\end{cases}
\]

Note that when $\Div_i > b$, changing any $\Div_i - b$ divisible items to indivisible does not decrease $\MMS_i$, since $\MMS_i = a + \frac{b}{n}$ remains unchanged. This transformation makes the allocation problem strictly harder, as divisible items can be treated as indivisible, but not vice versa. Therefore, for the purpose of analyzing worst-case approximation guarantees, it suffices to assume $\Div_i \leq b$ for all agents.



\medskip

We prove the theorem by analyzing three cases: Case 1 where $a=0$, Case 2 where $a=1$, and Case 3 where $a\ge 2$. 

\medskip

\textbf{Case 1:} $a=0$.
The problem is trivial if $b=0$. 
Thus, it is assumed that $0<b<n$ and $0<\MMS_i<1$ for every agent. 
For this case, we show that we can compute an MMS allocation.


The outline of the algorithm is as follows.
We initialize two sets of agents $N_c=N$ and $N_r=\emptyset$.  
Let $a_i \in N_c$ be the agent whose MMS is the smallest. 
If there is an entire item $g$ that is divisible for agent $a_i$, allocate $\len(g, \MMS_i)$ fraction of this item to $a_i$.
Otherwise, we check if the total value of all current fractional items in $M^{DIV}_i$ is at least $\MMS_i$; if so, allocate a minimum set of these items with value $\MMS_i$ to $a_i$.
In either case, agent $a_i$ is satisfied and exits the algorithm with her allocation.
If both attempts fail, we put $a_i$ in $N_r$ and move to the next round.
We repeat the above procedures until $N_c$ is empty. 
Then we give each agent in $N_r$ an entire item and allocate the remaining items arbitrarily.
The algorithm is formally described in Algorithm \ref{alg:una:n}.

\begin{algorithm}
\caption{Computing an MMS allocation for unary valuations when $m<n$.}\label{alg:una:n}
\KwIn{A unary instance $I = \langle{N,M, \fv}\rangle$ where $v_i(g) = 1$ for all $a_i\in N, g\in M$ and $m<n$.} 

\KwOut{An MMS allocation $(A_1,\ldots,A_n)$.}
 
Initialize $A_i \leftarrow \emptyset$, for all $a_i\in N$ and $N_c\leftarrow N$, $N_r\leftarrow\emptyset$.\\

\While{$N_c \neq \emptyset$}{
\label{step:alg:unary:case1}
Let $a_i \in \arg\min_{a_j\in N_c}\MMS_j$ be an agent who has the minimum MMS value.

\eIf{there exists an entire item $g\in M$ that is divisible for $a_i$}{
 Set $A_i\leftarrow \{\MMS_i\cdot g\}$, $M \leftarrow M \setminus A_i$, and $N_c\leftarrow N_c\setminus \{a_i\}$. \label{step:unary:case1}\\
}{
\eIf{ $v_i(M_i^\text{DIV})\geq \MMS_i$
}{
\While{$v_i(A_i)<\MMS_i$}{
\label{step:unary:case2}

 Let $g\in M_i^\text{DIV}$ be the remaining item with the smallest index that is divisible for agent $a_i$. 
 

 $A_i\leftarrow A_i \cup \{\min\{\len(g),\MMS_i-v_i(A_i)\}\cdot g\}$. \\
 
}
 Set $M \leftarrow M \setminus A_i$ and $N_c \leftarrow N_c\setminus \{a_i\}$.\\
}{
  Set $N_c \leftarrow N_c\setminus \{a_i\}$ and $N_r\leftarrow N_r\cup \{a_i\}$.  \\
}
}}

\For{$a_i\in N_r$}{
Choose an entire item $g\in M$ and set $A_i\leftarrow \{g\}$, $M \leftarrow M\setminus A_i$ and $N_r \leftarrow N_r \setminus a_i$. \label{step:unary:case3}\\
}
 Choose one agent $a_i\in N$ arbitrarily and set $A_i\leftarrow A_i\cup M$. 
 
\textbf{return} $(A_1, \ldots, A_n)$
\end{algorithm}


To prove the above algorithm always returns an MMS allocation, it suffices to show that after the removal of every agent, the remaining agents' MMS values in the reduced instance do not decrease.
Note that for an agent $a_i\in N$ removed and a set of items $S\subseteq M$ assigned to her, 
after the reduction, agent $a_j$'s MMS does not decrease if
\begin{align}
\label{eq:unary:case1:1}
    |(M\setminus S)_j^\text{IND}|+\lfloor \frac{v_j((M\setminus S)_j^\text{DIV})}{\MMS_j}\rfloor\geq |N|-1.
\end{align}
This is because in the reduced instance, each indivisible item can be entirely put into one bundle and the divisible items can be evenly partitioned so that each bundle has a value of at least $\MMS_j$.

We next prove by induction that Inequality \eqref{eq:unary:case1:1} holds after the removal of every agent. 
Before the execution of the algorithm, $|M_j^\text{IND}|+\lfloor \frac{v_j(M_j^\text{DIV})}{\MMS_j}\rfloor\geq |N|$ holds for every agent $a_j\in N$ by the definition of $\MMS_j$.
Suppose Inequality \eqref{eq:unary:case1:1} holds when the reduced instance contains $1< k < |N|$ agents. 
Let $N_c$ and $N_r$ (where $|N_c|+|N_r| = k$) be the sets of agents defined in Algorithm \ref{alg:una:n} at this moment and $a_i \in N_c$ has the smallest MMS value.
Let $a_j \in N\setminus\{a_i\}$ be any other agent.


Firstly, if there exists an entire item $g \in M$ which is divisible for agent $a_i$ (i.e., Step \ref{step:unary:case1}), 
we allocate $\MMS_i \cdot g$ to agent $a_i$, denoted as $A_i$.
If $g$ is divisible for agent $a_j \in N_c$, since $\MMS_i\leq \MMS_j$, 
\[
|(M\setminus A_i)_j^\text{IND}|+\lfloor\frac{v_j((M\setminus A_i)_j^\text{DIV})}{\MMS_j}\rfloor=|M_j^\text{IND}|+\lfloor\frac{v_j(M_j^\text{DIV})-v_j(A_i)}{\MMS_j}\rfloor\geq k-1.
\]
If $g$ is indivisible for $a_j$, 
\[
|(M\setminus A_i)_j^\text{IND}|+\lfloor\frac{v_j((M\setminus A_i)_j^\text{DIV})}{\MMS_j}\rfloor=|M_j^\text{IND}|-1+\lfloor\frac{v_j(M_j^\text{DIV})}{\MMS_j}\rfloor\geq k-1.
\]
For both cases, $\MMS_j$ is not decreased.

Secondly, if there is no entire divisible item in $M$ for agent $a_i$ and $v_i(M_i^\text{DIV})\geq \MMS_i$, we allocate a subset of items $A_i \subseteq M_i^\text{DIV}$ to agent $a_i$ such that $v_i(A_i)=\MMS_i$ (i.e., Step \ref{step:unary:case2}).
Since we only allocate the items which are not entire, the number of entire indivisible items remains the same after $a_i$ leaves.
Additionally, since $v_j(A_i\cap \cM_j^\text{DIV})\leq v_j(A_i)=v_i(A_i)= \MMS_i\leq \MMS_j$,
\[
|(M\setminus A_i)_j^\text{IND}|+\lfloor\frac{v_j(M\setminus A_i)_j^\text{DIV})}{\MMS_j}\rfloor\geq |M_j^\text{IND}|+ \lfloor\frac{v_j(M_j^\text{DIV})-\MMS_j}{\MMS_j}\rfloor\geq  k-1,
\]
which is as desired. 

Finally, if $a_i$ is not removed in the previous two cases, it will be added to $N_r$ for later allocation. 
Thus when $N_c = \emptyset$ and the algorithm reaches Step \ref{step:unary:case3}, for any $a_i \in N_r$, $v_i(M_i^\text{DIV})<\MMS_i$.
Therefore, $|M_i^\text{IND}|+\lfloor\frac{v_i(M_i^\text{DIV})}{\MMS_j}\rfloor\geq k$, which means $|M_i^\text{IND}| \ge k$, i.e., there are at least $k$ entire indivisible items in $M$. 
Allocating one entire item to every agent in $N_r$ and the rest items (if any) to an arbitrary agent ensures MMS, which proves this case.

\paragraph{Case 2: $a=1$.}
For any agent $a_i\in N$, we have $\MMS_i=1+\frac{\Div_i}{\Div_i+(n-b)}$, and hence $1\leq \MMS_i<2$. The main challenge is to satisfy agents for whom $\Div_i > n-b$, as each such agent must receive more than one item. Let $N_c=\{a_i \mid \Div_i>n-b\}$ denote the set of these \emph{critical agents}.

The algorithm for this case is summarized in Algorithm~\ref{alg:una:2n}.

\begin{algorithm}
\caption{Computing a $\frac{2}{3}$-MMS allocation for unary valuations when $n \le m < 2n$}\label{alg:una:2n}
\KwIn{An instance $I = \langle N, M, \fv \rangle$ where $v_i(g) = 1$ for all $a_i \in N$, $g \in M$, and $m = n + b < 2n$.}
\KwOut{A $\frac{2}{3}$-MMS allocation $(A_1,\ldots,A_n)$.}

Let $N_c = \{a_i \mid \Div_i > n - b\}$.

\If{$b \le \frac{n}{2}$}{
    Allocate items so that each agent receives at least one entire item.
}

\ElseIf{$b > \frac{n}{2}$ and $|N_c| \le b$}{
    Allocate two entire items to each agent in $N_c$ and at least one entire item to each agent in $N \setminus N_c$.
}

\ElseIf{$|N_c| > b > \frac{n}{2}$}{
    Compute a maximum 2-to-1 matching $(N_s, M_s)$ between $N_c$ and $M$, where an edge connects $a \in N_c$ and $g \in M$ if $g$ is divisible for $a$.

    Allocate as follows:
    \begin{itemize}
        \setlength{\itemsep}{-2pt}
        \item Each agent in $N_c \setminus N_s$ receives two entire items from $M \setminus M_s$.
        \item Each agent in $N_s$ receives half of her matched item in $M_s$ and one entire item from $M \setminus M_s$.
        \item Each agent in $N \setminus N_c$ receives one entire item from $M \setminus M_s$.
    \end{itemize}
    
    Allocate any remaining items arbitrarily.
}

\Return{$(A_1, \ldots, A_n)$}
\end{algorithm}

We distinguish three sub-cases.

\begin{itemize}[leftmargin=*]
    \item \textbf{Sub-Case 2.1:} $b \le \frac{n}{2}$.

    Since $\Div_i \le b$, we have $\MMS_i \le 1 + \frac{b}{n} \le \frac{3}{2}$ for all agents. It suffices to allocate each agent a single entire item to ensure the $\frac{2}{3}$ approximation. Any such allocation satisfies the required guarantee.

    \item \textbf{Sub-Case 2.2:} $b > \frac{n}{2}$ and $|N_c| \le b$.

    Recall that $\MMS_i = 1 + \frac{\Div_i}{\Div_i + n - b}$, and $\MMS_i > \frac{3}{2}$ if and only if $\Div_i > n - b$. Thus, only agents in $N_c$ have $\MMS_i > \frac{3}{2}$. If $|N_c| \le b$, then we have enough items to allocate two entire items to each agent in $N_c$ and one entire item to each agent in $N \setminus N_c$. Any remaining items can be allocated arbitrarily. This yields a $\frac{2}{3}$-MMS allocation.

    
    
    \item \textbf{Sub-Case 2.3:} $|N_c| > b > \frac{n}{2}$.
    
    This is the most challenging case, where there are not enough items to give two entire items to every critical agent. Our strategy is to identify a subset $N_s \subseteq N_c$ of agents who can each receive half of a divisible item and one entire item, while agents in $N_c \setminus N_s$ receive two entire items, and agents in $N \setminus N_c$ receive one entire item. This structure ensures a $\frac{2}{3}$-MMS allocation.
    
    We construct a bipartite graph $G = (N_c, M; E)$ where an edge connects $a_i \in N_c$ and $g \in M$ if $g$ is divisible for $a_i$. Each agent $a_i \in N_c$ has degree $\Div_i \ge n - b$. We consider triples $(a_i, a_j, g)$ whenever $g$ is divisible for both $a_i$ and $a_j$. We compute a maximum 2-to-1 matching $\Psi$ that matches agents $N_s \subseteq N_c$ and items $M_s \subseteq M$. 
    For simplicity, we omit the term ``2-to-1'' when referring to matchings. For each item $g$, let $\Psi(g)$ denote the pair of agents matched to it.
    
    \medskip
    
    We next use the following key lemma, which is proved after the main proof.
    
    \begin{lemma}\label{lem:unary:case2.3}
        $|M_s| \ge |N_c| - b$.
    \end{lemma}

    As we mentioned at the beginning of this sub-case, every pair of matched agents in $N_s$ evenly shares the matched item and receives one additional entire item. Each agent in $N_c \setminus N_s$ receives two entire items, and each agent in $N \setminus N_c$ receives one entire item. Thus, the total number of allocated items is
    \begin{align*}
        \frac{3}{2}\cdot |N_s| + 2\cdot (|N_c| - |N_s|) + (n - |N_c|) = |N_c| - |M_s| + n.
    \end{align*}
    By Lemma~\ref{lem:unary:case2.3}, $|M_s| \ge |N_c| - b$, so the total number of allocated items is at most $n + b = m$. Any remaining items can be allocated arbitrarily without affecting the approximation guarantee. 
\end{itemize}
Combining with the proofs of Cases 1 and 3 in the appendix, we have proved the theorem.
\end{proof}

\medskip

Finally, we provide the proof of Lemma~\ref{lem:unary:case2.3}.

\begin{proof}[Proof of Lemma~\ref{lem:unary:case2.3}]
We may assume without loss of generality that $\Div_i = n - b + 1$ for all agents $a_i \in N_c$. This assumption only reduces the number of divisible items and therefore weakens the matching potential. Therefore it is a valid worst-case reduction.

Suppose, for contradiction, that $|M_s| < |N_c| - b$. We partition the matched agents and items as follows. Initially set $N_1 = N_s$, $M_1 = M_s$, $N_2 = \emptyset$, $M_2 = \emptyset$. Let $N_3 = N_c \setminus N_s$ and $M_3 = M \setminus M_s$.
We iteratively perform the following operation: if there exists an item $g \in M_1$ that is divisible for two agents in $N_2 \cup N_3$, move both matched agents $\Psi(g) = \{a_1, a_2\}$ from $N_1$ to $N_2$ and move $g$ from $M_1$ to $M_2$.
Let $N_1^i, N_2^i, M_1^i, M_2^i$ be the state after $i$ such moves. We now make the following three claims.


\begin{claim}\label{cla:2}
For any $i \ge 1$, the items $\{g^1,\ldots,g^i\}$ can be matched with $2i$ agents in $(N_2^i \cup N_3) \setminus S$, where $S \subseteq N_2^i \cup N_3$ is any set of two agents.
\end{claim}

\begin{proof}[Proof of Claim \ref{cla:2}]
    We prove the claim by induction. 
    When $i =1$, item $g^1$ is divisible for four agents, namely, $a_1^1,a_2^1$, and two agents in $N_3$ who think $g^1$ is divisible.  
    Thus the removal of any two agents in $N_2^1\cup N_3$ still admits a match between two remaining agents and $g^1$.

    Assume the claim is true for $i = 1,\ldots, k-1$ and we consider the $k$-th round. 
    Let $b_1$ and $b_2$ be the two agents in $N_2^{k-1}\cup N_3$ for whom $g^k$ is divisible. 
    For any $S \subseteq N_2^k \cup N_3$ with $|S| = 2$, we construct a matching between $\{g^1,\ldots, g^k\}$ and $(N_2^k \cup N_3) \setminus S$ with cardinality $k$ as follows.

    If $S = \{a_1^k, a_2^k\}$, matching $g^k$ with $\{b_1,b_2\}$, together with a matching of cardinality $k-1$ between $\{g^1,\ldots, g^{k-1}\}$ and $(N_2^{k-1} \cup N_3) \setminus \{b_1,b_2\}$ gives the desired matching. 

    If $|S \cap \{a_1^k, a_2^k\} |= 1$, match $g^k$ with $S \cap \{a_1^k, a_2^k\}$ and one agent in $\{b_1,b_2\}\setminus S$.
    Note that exactly two items have been removed from $N_2^{k-1} \cup N_3$: one in $S$ and one matched to $g^k$. 
    Thus $\{g^1,\ldots, g^{k-1}\}$ and the remaining agents in $N_2^{k-1} \cup N_3$ admits a matching with cardinality $k-1$.

    If $S \cap \{a_1^k, a_2^k\} = \emptyset$,  match $g^k$ with $\{a_1^k, a_2^k\}$ together with a matching of cardinality $k-1$ between $\{g^1,\ldots, g^{k-1}\}$ and $(N_2^{k-1} \cup N_3) \setminus S$ gives the desired matching. 

    Combining the above three cases proves the claim. 
\end{proof}

\begin{claim}\label{divide}
No item in $M_1 \cup M_3$ is divisible to more than one agent in $N_2 \cup N_3$.
\end{claim}

\begin{proof}[Proof of Claim \ref{divide}]
    It is straightforward that no item in $M_1$ is divisible for more than one agent in $N_2\cup N_3$ since otherwise more agents will be moved from $N_1$ to $N_2$.  
    For any item $g \in M_3$, we prove the claim by contradiction. 
    Suppose $g$ is divisible to two agents $a,a'\in N_2\cup N_3$.
    If $N_2 = \emptyset$, which means $g$ is divisible to two agents in $N_3$, then this is a contradiction with $\Psi$ being maximum.
    If $N_2 \neq \emptyset$, by Claim \ref{cla:2}, all items in $M_2$ can be matched to $(N_2\cup N_3)\setminus\{a,a'\}$.
    Then, combining this matching with the matching between $M_1$ and $N_1$ and $(a,a',g)$, we find a larger matching than $\Psi$, which is also a contradiction with $\Psi$ being maximum.
\end{proof}

\begin{figure}
\centering
\includegraphics[width=0.45\textwidth]{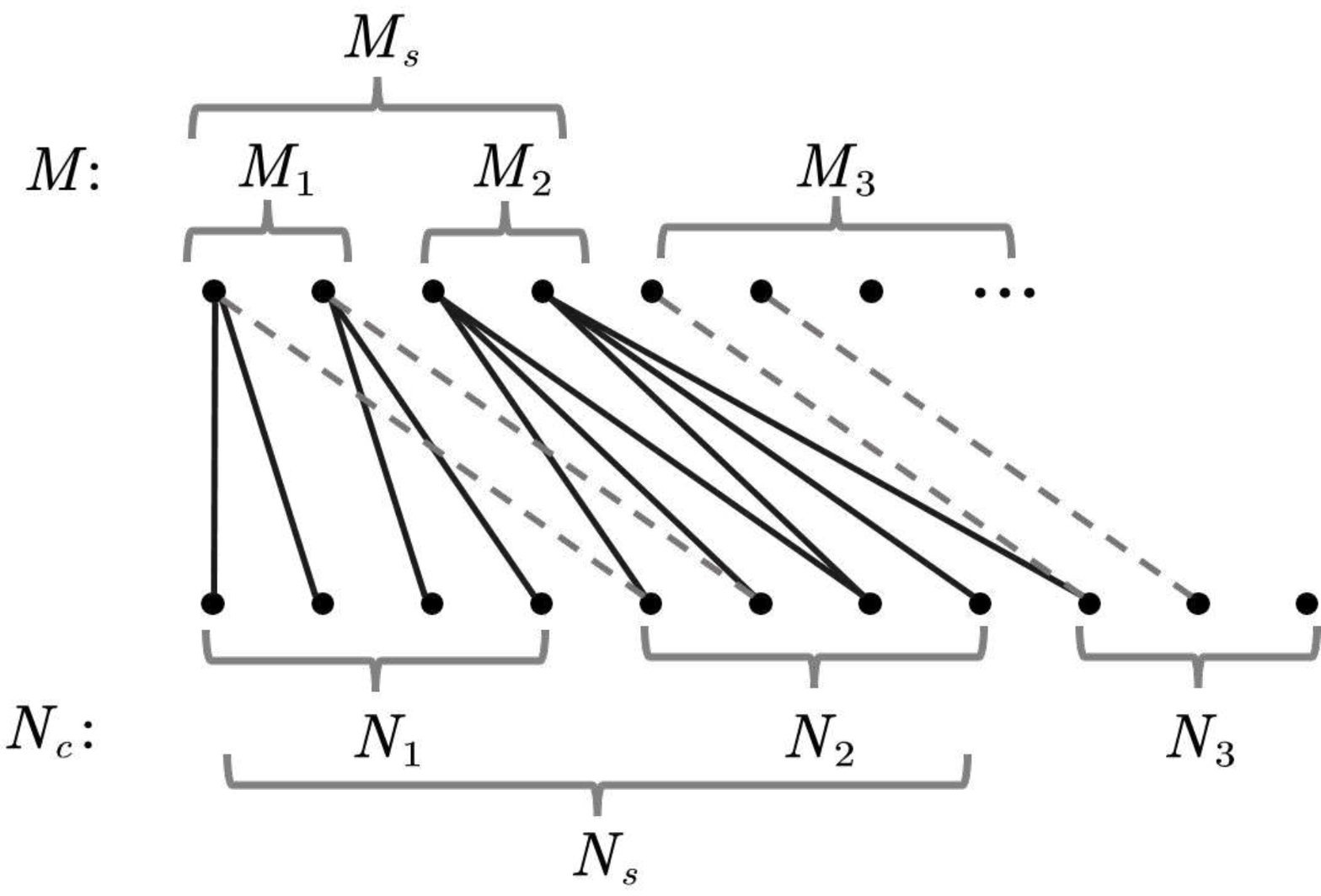}
\caption{The structure of $N_c$ and $M$ in Case 2 of Theorem \ref{the:una}}
\label{fig:unary:case2:division}
\end{figure}

By Claim~\ref{divide}, the total number of edges between $N_2 \cup N_3$ and $M_1 \cup M_3$ is at most $|M_1| + |M_3| = n + b - |M_2|$. Additionally, there are at most $|M_2| \cdot |N_2 \cup N_3|$ edges from $M_2$.

Combining these, the maximum number of edges between $N_2 \cup N_3$ and $M$ is bounded by the solution to the following integer program:
\begin{align}
\max & \quad |M_2| \cdot |N_2 \cup N_3| + n + b - |M_2| \notag \\
\text{s.t.} & \quad 0 < |M_2| \le |M_s|, \notag \\
& \quad |N_2 \cup N_3| \le |N_c|, \notag \\
& \quad \frac{n}{2} < b < |N_c| \le n, \notag \\
& \quad |M_s| < |N_c| - b. \notag
\end{align}

\begin{claim}\label{claim:unary:case2:3}
The optimal value of the above integer program is strictly less than $|N_2 \cup N_3| \cdot (n - b + 1)$.
\end{claim}

\begin{proof}[Proof of Claim \ref{claim:unary:case2:3}]
Consider the difference between the objective function and the target value in the claim,
\[
 (|M_2| \cdot |N_2\cup N_3|+n+b-|M_2|)- (|N_2\cup N_3|\cdot (n-b+1)).
\]
Since $|M_s|\leq |N_c|-b-1$, $|N_2|=2|M_2|$ and $|N_3|=|N_c|-|N_s|$, we have 
\[
|N_2\cup N_3|= |N_c|-|N_s|+2|M_2|\geq 2|M_2|+2b-|N_c|+2.
\]
Since $|M_2|-(n-b+1)<0$, the highest value of $|M_2| \cdot (|N_2\cup N_3|+n+b-|M_2|)-(|N_2\cup N_3|\cdot (n-b+1))$ is attained when  $|N_2\cup N_3|=2|M_2|+2b-|N_c|+2$, and thus
\[
\begin{aligned}
&|M_2| \cdot (|N_2\cup N_3|+n+b-|M_2|)-(|N_2\cup N_3|\cdot (n-b+1))\\
\leq& |M_2|\cdot (2|M_2|+2b-|N_c|+2)+n+b-|M_2|-(n-b+1)\cdot (2|M_2|+2b-|N_c|+2).
\end{aligned}
\]
Further, the highest value of the above term is attained when $|M_2|=|N_c|-b-1$,
\[
\begin{aligned}
&|M_2|\cdot (2|M_2|+2b-|N_c|+2)+n+b-|M_2|-(n-b+1)\cdot (2|M_2|+2b-|N_c|+2)\\
\leq & (|N_c|-b-1)\cdot  (2(|N_c|-b-1)+2b-|N_c|+2)+n+b-(|N_c|-b-1)\\
&-(n-b+1)\cdot(2|M_2|+2b-|N_c|+2)\\
=&|N_c|^2-3|N_c|+n+2b+1-n\cdot |N_c|.
\end{aligned}
\]
For $|N_c|^2-3|N_c|+n+2b+1-n\cdot |N_c|$, the highest value is attained when $|N_c|=n$,
\[
\begin{aligned}
|N_c|^2-3|N_c|+n+2b+1-n\cdot |N_c|\leq n^2-3n+n+2b+1-n^2=2b+1-2n\leq -1.
\end{aligned}
\]
Hence 
\[
(|M_2| \cdot |N_2\cup N_3|+n+b-|M_2|)-(|N_2\cup N_3|\cdot (n-b+1))\leq -1,
\]
which proves the claim.
\end{proof}

On the other hand, since each agent in $N_c$ has degree at least $n - b + 1$, the total number of edges must be at least $|N_2 \cup N_3| \cdot (n - b + 1)$. This contradicts Claim~\ref{claim:unary:case2:3}, and thus proves the lemma.
\end{proof}

\textbf{Case 3:} $a\geq 2$.
In this case, for any agent $a_i\in N$, $2\le  a\leq \MMS_i<a+1$.
Thus we can arbitrarily allocate items as long as every agent receives at least $a$ items, and the approximation ratio is at least $\frac{a}{a+1}\ge \frac{2}{3}$.

\section{Subjective Divisibility with Heterogeneous Valuations}
\label{sec:main}

This section presents our main result for the general setting in which agents have both subjective views on item divisibility and heterogeneous valuations. We show that for any such instance, there always exists a $\frac{5}{9}$-MMS allocation, improving the $\frac{1}{2}$-approximation previously obtained by~\citet{DBLP:journals/corr/abs-2310-00976}.

\begin{theorem}\label{the:gen}
    For any subjective divisibility instance, there exists a $\frac{5}{9}$-MMS allocation. 
\end{theorem}

Our proof is constructive. Given any instance, we explicitly construct a $\frac{5}{9}$-MMS allocation through three phases detailed in the subsections below. We remark that if the MMS values are given, the algorithm runs in polynomial time.

\subsection{Phase 1: Preprocess Large and Sharable Items}

Before designing the main algorithm, we preprocess the instance $I=\langle{N, M,\fv}\rangle$ and introduce notations used later. For simplicity, we assume $\MMS_i=1$ for every agent $a_i \in N$. 
We define the following types of items:
\begin{itemize}[leftmargin=*]
    \item An item $g \in M$ is \textit{large} for agent $a_i$ if $v_i(g) \ge \frac{5}{9}$.
    \item An item $g \in M$ is \textit{medium} for agent $a_i$ if $v_i(g) \ge \frac{7}{18}$.
    \item An item $g$ is \textit{sharable} for agent $a_i$ if it is both medium and divisible for $a_i$, i.e., $v_i(g) \ge \frac{7}{18}$ and $g$ is divisible for $a_i$.
\end{itemize}

\paragraph{Step 1}
We process all large items. If a large item $g$ is divisible for some agent $a_i$ with $v_i(g) \ge \frac{5}{9}$, we allocate the minimum fraction of $g$ that gives value at least $\frac{5}{9}$ to the agent with the highest $v_i(g)$. If $g$ is indivisible for all such agents, we assign it entirely to one of them. In either case, the selected agent receives at least $\frac{5}{9}$ and is removed from the instance along with her allocated bundle. This step is repeated until no large item remains.

\medskip

We remark that, unlike the standard MMS setting, removing a large item may decrease the MMS values of the remaining agents under subjective divisibility. 
Thus, we cannot simply reduce the instance and assume without loss of generality that no large item remains. 
Instead, in all subsequent steps, we continue to work with the original instance and MMS values refer to those in the original instance.

Let $I'=\langle N', M', \fv \rangle$ be the residual instance after Step 1, and let $u = |N'|$. 
Then we have
\begin{align}\label{eq:general:phase1:M'}
    v_i(M') \ge u
\end{align}
for every $a_i \in N'$, and $v_i(g) < \frac{5}{9}$ for all $g \in M'$.

Note that Inequality~\eqref{eq:general:phase1:M'} is not tight. 
This is because each remaining agent may have lost divisible items partially allocated to others, with lost value at most $\frac{5}{9}$ per item.
We will later refine this inequality when a tighter bound is required in the analysis.

Although some items may have been partially allocated, we continue to treat each remaining item as {\em entire} and normalize its length to 1.

\paragraph{Step 2}
We next allocate \textit{sharable} items to pairs of agents. We construct a bipartite graph $G = (N', M'; E)$ where the vertices are agents $N'$ and items $M'$. There is an edge between $a_i$ and $g$ if $g$ is sharable for $a_i$.

We search for disjoint triples $(a_i, a_j, g)$ in which $g$ is sharable to both $a_i$ and $a_j$. Fix a maximum 2-to-1 matching $\Psi$ in $G$, which pairs items $M_X \subseteq M'$ with agent pairs $X \subseteq N'$. For each $g \in M_X$, denote by $\Psi(g)$ the pair of agents matched to $g$.

We divide each item $g \in M_X$ evenly between its two matched agents. Each such agent receives value at least $\frac{7}{36}$ from her share. We refer to each half of a shared item as a \emph{base}. The allocation of bases is final and will not be changed in subsequent steps.

For any agent not matched in $\Psi$, the value of each base is at most $\frac{5}{18}$ due to the exclusion of large items in Step 1.

\subsection{Phase 2: Main Algorithm and Processing Non-critical Agents}

Let $M_R = M' \setminus M_X$ and $k = |X|$. Without loss of generality, assume $|M_R| \ge 2u - k$, as otherwise we can add dummy items valued at 0 by all agents.

We classify the agents in $N' \setminus X$ into two subsets $Y$ and $Z$. For each agent $a_i \in N'$, let $M_R^i[l:r]$ denote the set of items ranked from $l$-th to $r$-th largest by $v_i$. Define $Z$ to be the set of agents $a_i \in N' \setminus X$ satisfying both:
\begin{itemize}
    \item Agent $a_i$ has at least $\lceil \frac{4u}{3} \rceil$ medium items in $M_R^i[1:2u-k]$;
    \item Agent $a_i$ has at least 5 shareable items in $M_R^i[1:2u-k]$.
\end{itemize}
Agents in $Z$ are called {\em critical agents}. We assume $|Z| \le \lceil \frac{u}{3} \rceil$ and let $Y = N' \setminus (X \cup Z)$. The case $|Z| > \lceil \frac{u}{3} \rceil$ is deferred to Phase 3.

An instance is said to have {\em identical ordering} (IDO) if there exists a common item order $g_1, g_2, \ldots, g_m$ such that $v_i(g_1) \ge v_i(g_2) \ge \cdots \ge v_i(g_m)$ for all $i$. Although IDO reduction is not valid in the subjective setting, it remains useful when analyzing allocations where no item is divided. 

We obtain an IDO instance $I'' = (N', M_R, \fv')$ from $I' = (N', M_R, \fv)$ by fixing an item order and sorting each $v_i$ accordingly in decreasing order.
That is, in each $v'_i$, the $j$-th largest value is the value for the $j$-th item. Since we do not modify agents' divisibility views and will only allocate whole items, this reduction is without loss of generality. We formalize this in the following claim:

\begin{claim}\label{claim:general:ido}
Let $\alpha_1,\ldots,\alpha_u$ be arbitrary nonnegative values. If there is an allocation $\A$ in $I''$ that assigns each item entirely to one agent and satisfies $v'_i(A_i) \ge \alpha_i$ for all $i$, then we can compute in polynomial time an allocation $\B$ in $I'$ such that each item is again entirely allocated, and $v_i(B_i) \ge \alpha_i$ for all $i$.
\end{claim}

The indivisible version of this claim appears in \citep{DBLP:journals/teco/BarmanK20}. A similar argument proves our claim since all items are allocated entirely.

\medskip

By Claim~\ref{claim:general:ido}, we may assume without loss of generality that $I'$ is IDO as long as all items in $M_R$ are allocated wholly. We thus omit the superscript $i$ in $M_R^i[l:r]$. The algorithm for this phase, shown in Algorithm~\ref{alg:general}, consists of two main steps: selection and bag-filling.


\begin{algorithm}[t]
\caption{Allocating items in Phase 2: when $|Z| \le \lceil \frac{u}{3} \rceil$}
\label{alg:general}
\textbf{Input:} IDO instance $\langle N', M_R, \fv \rangle$, base items $M_X$, picking order $\sigma = (X, Y, Z)$\\
\textbf{Output:} Allocation $\A = (B_1, \ldots, B_u)$

\vspace{1mm}
\textit{\% Process base items} \\
For each item $g \in M_X$, allocate $g$ evenly to $\Psi(g) = \{a_i, a_j\}$, denoted by $g_i^b$ and $g_j^b$

\vspace{1mm}
\textit{\% Selection step} \\
Initialize $A_i = B_i = \emptyset$ for all $a_i \in N'$\\
\For{$i = 1, \ldots, u - 1, u, u, u - 1, \ldots, k + 1$}{
    Let $g' \in \arg\max_{g \in M_R} v_i(g)$ \\
    Set $A_i \leftarrow A_i \cup \{g'\}$ and $M_R \leftarrow M_R \setminus \{g'\}$
}

\If{there exists $a_i \in X$ with $v_i(A_i) \ge \frac{5}{9}$}{
    Set $B_i \leftarrow A_i \cup \{g_i^b\}$ and $X \leftarrow X \setminus \{a_i\}$
}
\If{there exists $a_i \in Y$ with $v_i(A_i) \ge \frac{5}{9}$}{
    Set $B_i \leftarrow A_i$ and $Y \leftarrow Y \setminus \{a_i\}$
}

\vspace{1mm}
\textit{\% Bag-filling step} \\
Let $\mathcal{B} = \{A_i\}_{a_i \in X \cup Y}$ \\
\While{$X \cup Y \ne \emptyset$}{
    Select an arbitrary bag $B \in \mathcal{B}$\\
    \While{no $a_i \in Y$ with $v_i(B) \ge \frac{5}{9}$ and no $a_i \in X$ with $v_i(B) \ge \frac{13}{36}$}{
        Choose any $g \in M_R$, set $B \leftarrow B \cup \{g\}$, $M_R \leftarrow M_R \setminus \{g\}$
    }
    \uIf{there exists $a_i \in Y$ with $v_i(B) \ge \frac{5}{9}$}{
        Set $B_i \leftarrow B$ and $Y \leftarrow Y \setminus \{a_i\}$
    }
    \ElseIf{there exists $a_i \in X$ with $v_i(B) \ge \frac{13}{36}$}{
        Set $B_i \leftarrow B \cup \{g_i^b\}$ and $X \leftarrow X \setminus \{a_i\}$
    }
}

\If{$M_R \ne \emptyset$}{
    Arbitrarily allocate remaining items in $M_R$ to any agents
}

\textbf{return} $(B_1, \ldots, B_u)$
\end{algorithm}

\paragraph{Step 3: Selection}
Fix an ordering of agents $\sigma = (X, Y, Z)$. The ordering within $X$ and $Y$ can be arbitrary; the order within $Z$ will be specified in Phase 3. Let $\A = (A_1, \ldots, A_u)$ denote the allocation, initialized as empty.

We first allocate the largest $u$ items $\{g_1, \ldots, g_u\}$ in $M_R$, assigning $g_i$ to agent $a_i$ for $i = 1, \ldots, u$ (referred to as round 1). We then allocate the next $u - k$ largest items to agents in $Y \cup Z$ following the reverse order of $\sigma$. That is, agent $a_i$ receives $g_{2u - i + 1}$ for $i = k+1, \ldots, u$ (referred to as round 2).

After this step, we have $A_i = \{g_i\}$ for $i = 1, \ldots, k$, and $A_i = \{g_i, g_{2u - i + 1}\}$ for $i = k+1, \ldots, u$. We remark that these bundles are provisional and may be modified in Step 4.

By the definition of $Z$, the first $\lceil \frac{4u}{3} \rceil$ items in $M_R$ are medium for every agent in $Z$, and each such item is valued at least $\frac{7}{18}$. Hence, each agent in $Z$ receives two medium items. This gives the following claim. Note that we will provide a refinement of Claim \ref{claim:general:phase2:Z} in the third phase when $|Z| > \lceil \frac{u}{3} \rceil$.

\begin{claim}\label{claim:general:phase2:Z}
    If $|Z| \le \lceil \frac{u}{3} \rceil$, then $v_i(A_i) \ge \frac{7}{9}$ for all $a_i \in Z$.
\end{claim}

We thus finalize the allocations for agents in $Z$, and for those in $X \cup Y$ whose bundles already have value at least $\frac{5}{9}$. These agents will not participate in the next step. Let $X'$ and $Y'$ denote the remaining unsatisfied agents in $X$ and $Y$, respectively.

\paragraph{Step 4: Bag-filling}
We allocate the remaining items in $M_R$ (excluding the first $2u-k$ items) in a bag-filling manner to agents in $X'$ and $Y'$.
Let ${\mathcal B} = \{A_i\}_{a_i \in X'\cup Y'}$ be the set of current bundles.

\begin{itemize}\itemsep0em
    \item Arbitrarily select a bag $B\in {\mathcal B}$.
    \item Add items from $M_R$ to $B$ one at a time until either:
    \begin{itemize}\itemsep0em
        \item some agent $a_i\in X'$ satisfies $v_i(B)\ge \frac{13}{36}$, or
        \item some agent $a_i\in Y'$ satisfies $v_i(B)\ge \frac{5}{9}$.
    \end{itemize}
    \item Among the satisfied agents, select one $a_i$, breaking ties in favor of $Y'$ over $X'$.
    \item Assign $B$ to $a_i$, denoted by $B_i = B$. If $a_i\in X'$, also allocate $g^b_i$ to $a_i$.
    \item Remove $a_i$ and $B$ from the instance and repeat the process.
\end{itemize}

We next argue that every agent in $N'$ receives a bundle in Step 3 or Step 4 that meets the $\frac{5}{9}$-MMS threshold. It suffices to show that the final remaining agent $a_i$ can also be satisfied. 

\begin{lemma}
\label{lem:main:phase2:X}
If $a_i \in X$, then $a_i$ receives a bundle with value at least $\frac{5}{9}$ at the end of Step 4.
\end{lemma}

\begin{proof}
We first claim that $v_i(g_r)<\frac{13}{36}$ for all $r>k$; 
Otherwise 
\[
v_i(A_i) =v_i(g^b_i)+v_i(g_i) \ge \frac{1}{2}\cdot\frac{7}{18}+\frac{13}{36} = \frac{5}{9},
\]
which means $a_i$ should be removed in Step 3.
Let $a_j\neq a_i$ be any agent that is removed before $a_i$.

If $a_j \in X\setminus X'$, then 
\[
v_i(\{g^b_j\} \cup A_j) = v_i(g^b_j) + v_i(g_j) \le \frac{1}{2}\cdot\frac{5}{9}+\frac{5}{9} = \frac{5}{6}.
\]

If $a_j \in (Y\setminus Y')\cup Z$, then 
\[
v_i(A_j) = v_i(g_j)+v_i(g_{2u-j+1}) \le 2\cdot\frac{13}{36} = \frac{13}{18}.
\]

If $a_j \in X'\cup Y'$, let $e$ be the last item added to $B_j$. We have $v_i(B_j\setminus\{e\}) < \frac{13}{36}$ and $v_i(e)< \frac{13}{36}$. Then 
\begin{align*}
    v_i(\{g^b_j\} \cup B_j)&\le \frac{1}{2}\cdot\frac{5}{9}+ \frac{13}{36}+ \frac{13}{36} = 1, &\text{if $a_j \in X'$~}\\
    v_i(B_j) &\le \frac{13}{36}+ \frac{13}{36} = \frac{13}{18}, & \text{if $a_j \in Y'$}.
\end{align*}

Combining all the above cases,  
\[
\sum_{a_j\in (X\setminus X')\cup (Y\setminus Y') \cup Z} v_i(A_j)+\sum_{a_j\in X'\cup Y'} v_i(B_j) \le u-1,
\]
and thus after the removal of $u-1$ agents, there are enough remaining items for $a_i$ to get a bundle with value of $\frac{5}{9}$ or above, which proves the lemma.
\end{proof}

\begin{lemma}
\label{lem:main:phase2:Y}
If $a_i \in Y$, then $a_i$ receives a bundle with value at least $\frac{5}{9}$ at the end of Step 4.
\end{lemma}

\begin{proof}
We first claim that $v_i(g_r)<\frac{5}{18}$ for all $r > 2u - i + 1$; otherwise, $v_i(A_i) \ge 2v_i(g_{2u - i + 1}) \ge \frac{5}{9}$, so agent $a_i$ would have been removed in Step 3. 
Let $a_j \ne a_i$ be any agent removed before $a_i$. Since $a_i \notin Z$, there are at most $\lfloor \frac{4u}{3} \rfloor$ medium items or at most $4$ shareable items in $M_R[1:2u-k]$ for $a_i$.

We consider two cases. In Case 1, there are at most $\lfloor \frac{4u}{3} \rfloor$ medium items for $a_i$ in $M_R[1:2u-k]$. In Case 2, there are more than $\lfloor \frac{4u}{3} \rfloor$ medium items but at most $4$ shareable ones. We begin with Case 1.

If $a_j \in X \setminus X'$, then
\[
v_i(\{g^b_j\} \cup A_j) = v_i(g^b_j) + v_i(g_j) \le \frac{1}{2} \cdot \frac{5}{9} + \frac{5}{9} = \frac{5}{6}.
\]

If $a_j \in (Y \setminus Y') \cup Z$, then the bundle $A_j$ may contain two medium items. By assumption, there are at most $\lfloor \frac{u}{3} \rfloor$ such bundles, and
\begin{align}\label{eq:general:phase2:higher}
v_i(A_j) = v_i(g_j) + v_i(g_{2u - j + 1}) \le 2 \cdot \frac{5}{9} = \frac{10}{9}.
\end{align}

For the remaining bundles,
\begin{align}\label{eq:general:phase2:lower}
v_i(A_j) = v_i(g_j) + v_i(g_{2u - j + 1}) \le \frac{5}{9} + \frac{7}{18} = \frac{17}{18}.
\end{align}
Note that in Inequality~\eqref{eq:general:phase2:higher}, agent $a_j$ may take a bundle valued above $1$ for $a_i$. We must ensure the number of such cases is small enough that sufficient value remains for $a_i$. In all other cases, including Inequality~\eqref{eq:general:phase2:lower}, the taken bundle is valued at most $\frac{17}{18}$ by $a_i$.

If $a_j \in X'$, then $v_i(B_j) < \frac{5}{9}$ since $a_i$ has higher priority to exit. Thus,
\[
v_i(\{g^b_j\} \cup B_j) = v_i(g^b_j) + v_i(B_j) \le \frac{1}{2} \cdot \frac{5}{9} + \frac{5}{9} = \frac{5}{6}.
\]

If $a_j \in Y'$, let $e$ be the last item added to $B_j$. Then $v_i(B_j \setminus \{e\}) < \frac{5}{9}$ and $v_i(e) < \frac{5}{18}$, so
\[
v_i(B_j) = v_i(B_j \setminus \{e\}) + v_i(e) \le \frac{5}{9} + \frac{5}{18} = \frac{5}{6}.
\]

Combining all cases, and noting that only bundles under Inequality~\eqref{eq:general:phase2:higher} may exceed value $1$, we obtain
\begin{align*}
\sum_{a_j \in (X \setminus X') \cup (Y \setminus Y') \cup Z} v_i(A_j) + \sum_{a_j \in X' \cup Y'} v_i(B_j)
&\le \frac{10}{9} \cdot \left\lfloor \frac{u}{3} \right\rfloor + \frac{17}{18} \cdot (u - 1 - \left\lfloor \frac{u}{3} \right\rfloor) \\
&= \frac{17}{18}(u - 1) + \frac{3}{18} \cdot \left\lfloor \frac{u}{3} \right\rfloor \\
&\le u - 1 + \frac{1}{18}.
\end{align*}
Hence, after the removal of $u - 1$ agents, agent $a_i$ values the remaining items at least $\frac{17}{18}$, which suffices.

\medskip

We next consider Case 2, where there are more than $\lfloor \frac{4u}{3} \rfloor$ medium items for $a_i$ in $M_R[1:2u-k]$ and at most $4$ of them are divisible (so at most four items are shareable). 
In this case, all previous inequalities except \eqref{eq:general:phase2:higher} still hold: each agent $a_j$ takes away a bundle valued no more than $\frac{17}{18}$ by $a_i$. 
However, Inequality~\eqref{eq:general:phase2:higher} may now occur more than $\lfloor \frac{u}{3} \rfloor$ times, so the prior bound does not apply directly.

This is the most technical part of this proof.
To proceed, we need to revisit Step 1 and refine the lower bound on $v_i(M')$. 
Suppose $p_i$ agents were allocated indivisible items of $a_i$ and $q_i$ agents were allocated divisible items. 
Then $n = u + p_i + q_i$. 
Let $P_i$ and $Q_i$ be the sets of items allocated to these $p_i$ and $q_i$ agents, respectively.
Since $P_i$ contains $p_i$ indivisible items, the removal of each such item with one agent does not decrease $\MMS_i$ (see e.g.,~\citep{DBLP:conf/sigecom/GhodsiHSSY18,DBLP:conf/soda/GargMT19}). Hence,
\begin{align}\label{eq:general:phase2:Y2:1}
    v_i(M \setminus P_i) \ge (n - p_i) \cdot \MMS_i = u + q_i.
\end{align}
For divisible items in $Q_i$, which may be fractionally allocated, each fractional bundle is valued at most $\frac{5}{9}$ by $a_i$. Thus,
\begin{align}\label{eq:general:phase2:Y2:2}
    v_i(Q_i) \le \frac{5}{9} \cdot q_i.
\end{align}
Combining \eqref{eq:general:phase2:Y2:1} and \eqref{eq:general:phase2:Y2:2}, we obtain a stronger version of Inequality~\eqref{eq:general:phase1:M'}:
\begin{align}\label{eq:general:phase1:M':revised1}
v_i(M') \ge u + q_i - \frac{5}{9}q_i = u + \frac{4}{9}q_i.
\end{align}

\begin{figure}
\centering
\includegraphics[width=0.5\textwidth]{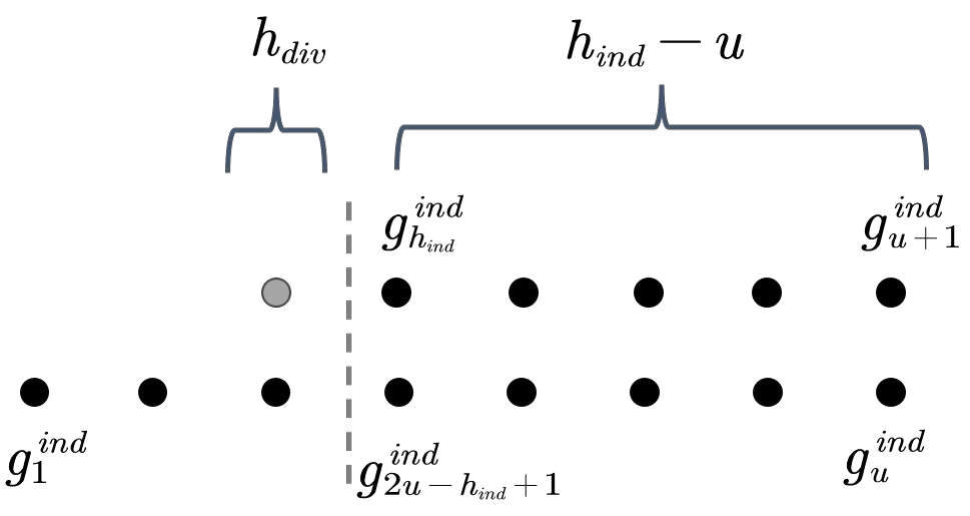}
\caption{The allocation of medium items of $a_i$}
\label{fig:selective}
\end{figure}

Next, let $L$ be the set of medium items in $M_R[1:2u-k]$, and let $L_i^{\text{DIV}}$ and $L_i^{\text{IND}}$ denote the divisible and indivisible items in $L$ for $a_i$. 
Set $h_{\text{div}} = |L_i^{\text{DIV}}|$, $h_{\text{ind}} = |L_i^{\text{IND}}|$, and note $|L| = h_{\text{div}} + h_{\text{ind}} \le 2u - k$.

For ease of presentation, we further label the items so that $L_i^{\text{DIV}} = \{g_1^{\text{div}}, \ldots, g_{h_{\text{div}}}^{\text{div}}\}$ and $L_i^{\text{IND}} = \{g_1^{\text{ind}}, \ldots, g_{h_{\text{ind}}}^{\text{ind}}\}$ in decreasing value order for $a_i$.
If $h_{\text{ind}} \le u$, the analysis from Case 1 still applies, which gives
\[
\sum_{a_j \in (X \setminus X') \cup (Y \setminus Y') \cup Z} v_i(A_j) + \sum_{a_j \in X' \cup Y'} v_i(B_j) \le u - 1 + \frac{1}{18}.
\]
We omit the details here.

Now suppose $h_{\text{ind}} > u$. In Step 3, items $\{g_{2u - h_{\text{ind}} + 1}^{\text{ind}}, \ldots, g_{h_{\text{ind}}}^{\text{ind}}\}$ are assigned in bundles of two (see Figure~\ref{fig:selective}). 
There are $h_{\text{ind}} + h_{\text{div}} - u$ such bundles. 
These are the only bundles that may exceed value 1 for $a_i$.

If $h_{\text{ind}} \le u + q_i$, then at most $h_{\text{div}} + h_{\text{ind}} - u$ bundles are valued above 1. 
The remaining $u - (h_{\text{div}} + h_{\text{ind}} - u) - 1$ bundles are each valued at most $\frac{17}{18}$. 
Using \eqref{eq:general:phase1:M':revised1}, we have
\[
    \underbrace{v_i(M')}_{\ge u+ \frac{4}{9}q_i} -\frac{10}{9}\cdot\underbrace{(h_{\text{div}}+h_{\text{ind}}-u)}_{\text{\# bundles with value $\le\frac{10}{9}$}}-
    \frac{17}{18} \cdot\underbrace{(u-(h_{\text{div}}+h_{\text{ind}}-u)-1)}_{\text{\# bundles with value $\le\frac{17}{18}$}} \geq  \frac{5}{9}.
\]

If $h_{\text{ind}} > u + q_i$, define $h_i = u + q_i$ and consider an augmented instance with $h_i$ agents, including the $q_i$ agents who received divisible items.

We have the following observation.

\begin{observation}\label{ob:general:phase2}
    For any allocation of $n + t$ indivisible items among $n$ agents, with $1 \le t < n$, the MMS value is at most the value of the $(n - t)$-th largest item.
\end{observation}

\begin{proof}
    In any $n$-partition of the $n+t$ indivisible items, if the $(n-t)$-th largest item is the only item in the bundle containing it, we have the claim in the observation.
    If there are at least two items in that bundle, then there must be another item between the $(n-t+1)$-th and $n$-th largest that alone occupies a bundle, which implies an even lower minimum value of all bundles.
\end{proof}

By Observation~\ref{ob:general:phase2}, removing the smallest $2t$ items and $t$ agents preserves MMS values.

\begin{figure}
\centering
\includegraphics[width=0.6\textwidth]{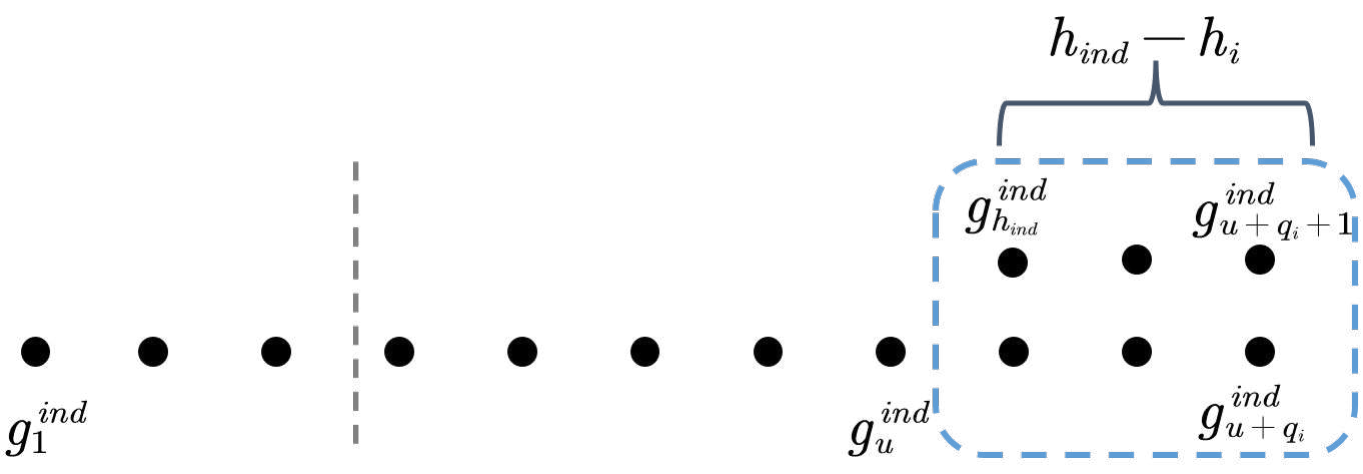}
\caption{The MMS partition of $a_i$ in the instance with $u+q_i$ agents}
\label{round}
\end{figure}

Apply this with $t = h_{\text{ind}} - h_i$ and remove $T$, the last $2(h_{\text{ind}} - h_i)$ indivisible items in value order. Then, by extending \eqref{eq:general:phase2:Y2:1}, we get
\begin{align}\label{eq:general:phase2:Y2:3}
v_i(M \setminus P_i \setminus T) \ge 2h_i - h_{\text{ind}}.
\end{align}

The remaining value for $a_i$ satisfies
\[
    v_i(M\setminus P_i \setminus T)-\underbrace{\frac{5}{9} \cdot q_i}_{\text{step 1}} - \frac{10}{9} \cdot \underbrace{(q_i+ h_{\text{div}})}_{\text{\# bundles with value $\le \frac{10}{9}$}} -\frac{17}{18}\cdot \underbrace{(u-(h_{\text{ind}}+h_{\text{div}}-u)-1)}_{\text{\# bundles with value $\le \frac{17}{18}$}}  \geq \frac{5}{9}.
\]

This completes the proof that $a_i$ receives a bundle worth at least $\frac{5}{9}$ in all cases.
\end{proof}

\subsection{Phase 3: Processing Critical Agents}

So far, under no conditions, agents in $X$ and $Y$ have received a bundle with value at least $\frac{5}{9}$. 
By Claim~\ref{claim:general:phase2:Z}, if $|Z| \le \lceil \frac{u}{3} \rceil$, then every agent in $Z$ can receive two medium items and be satisfied. 
However, this may not hold if $|Z| > \lceil \frac{u}{3} \rceil$.

In this phase, we show that each agent in $Z$ can still be allocated two medium items. 
Since the IDO assumption in Phase 2 changes the divisibility structure, we return to the original instance.
By definition, each agent $a_i \in Z$ has at least $\lceil \frac{4u}{3} \rceil$ medium items in $M_R^i[1:2u-k]$ and at least five of them are shareable.

Intuitively, Phase 1 maximizes the number of agents in $X$ who can share divisible items. 
Agents in $Z$ were not matched in that phase, so their medium items remain largely disjoint. 
We leverage this to construct a feasible allocation.

Recall the IDO reduction property shown in~\citep{DBLP:journals/teco/BarmanK20}: the allocation in the IDO instance induces a picking sequence in the original instance. 
If agents pick their favorite remaining items in this order, their values are preserved.
In Algorithm~\ref{alg:general}, the first $u - |Z|$ items are allocated to agents in $X \cup Y$, and the next $2|Z|$ items to agents in $Z$. 
Hence, in the original instance, agents in $X$ and $Y$ select first, after which the agents in $Z$ can receive in total $2|Z|$ items before $X$ and $Y$ select again.
It suffices to show that each $a_i \in Z$ can choose two medium items from $M_R^i[u - |Z| + 1, m]$ in the original instance.
Moreover, the values of the bundle allocated to agents in $X\cup Y$ are not affected.

\begin{lemma}
In the original instance, every agent $a_i \in Z$ can receive two medium items with total value at least $\frac{7}{9}$.
\end{lemma}

\begin{proof}
Fix an arbitrary order $\sigma$ of agents in $Z$. 
Each agent selects one favorite available item in order $\sigma$, and a second one in the reverse order.

Let $Z_1$ be the last $\lceil \frac{u}{3} \rceil$ agents in $\sigma$ and $Z_2 = Z \setminus Z_1$. 
By definition, each $a_i \in Z$ has at least $\lceil \frac{4u}{3} \rceil$ medium items. 
Thus, all agents in $Z_1$ get two medium items, and each agent in $Z_2$ gets at least one.

Let $Z_{2,1} \subseteq Z_2$ be agents who receive only one medium item, and $Z_{2,2} = Z_2 \setminus Z_{2,1}$ those who receive two. 
Reorder $Z_2$ in $\sigma$ so that agents in $Z_{2,1}$ appear first.

Let $Z_3 \subseteq Z_1 \cup Z_{2,2}$ be agents who have additional medium items beyond the ones used so far. 
Let $z_1 = |Z_{2,1}|$, $z_2 = |Z_3|$, and $z_3 = |Z_{2,2}|$.

Since no item is shareable to more than one agent, we have:
\[
5z_1 + 5(\lceil \tfrac{u}{3} \rceil + z_3 - z_2) \le \lceil \tfrac{4u}{3} \rceil + z_3,
\]
which simplifies to:
\[
z_1 + \tfrac{u}{15} + \tfrac{4}{5}z_3 \le z_2.
\]
We now pair each $a_i \in Z_{2,1}$ with a distinct agent $a_i' \in Z_3$ who has a medium item that $a_i$ values, and who still has an unused medium item available. 
This ensures each $a_i$ in $Z_{2,1}$ obtains a second medium item.

To see that such pairing is feasible, suppose not. 
Then some $a_i \in Z_{2,1}$ sees at least $\lceil \tfrac{u}{15} \rceil + \tfrac{4}{5}z_3$ agents exit without contributing a usable medium item to her. 
This implies at least $\lceil \tfrac{2}{15}z_3 \rceil + \tfrac{8}{5}z_3$ such items are unavailable to $a_i$.
But this would contradict the assumption that $a_i$ has $\lceil \tfrac{4u}{3} \rceil$ medium items, as
\[
\lceil \tfrac{4u}{3} \rceil + z_3 - \left(\lceil \tfrac{2}{15}u \rceil + \tfrac{8}{5}z_3 \right) \le \lceil \tfrac{4u}{3} \rceil
\]
would imply some unallocated medium item remains for $a_i$, a contradiction.

Thus, each $a_i \in Z$ can obtain two medium items.
\end{proof}

\section{Tight Approximation for Three and Four Agents}

In this section, we show that when there are at most four agents, a $\frac{2}{3}$-approximate MMS fair allocation can be computed in polynomial time.  
The detailed proofs and further discussion are provided in Appendix \ref{ape:3agent}.

It is noteworthy that for the cases of two and three agents, \citet{DBLP:journals/corr/abs-2310-00976} have proved the existence of a $\frac{2}{3}$-approximate MMS fair allocation.  
Our results achieve the same approximation guarantee through alternative methods, but our algorithms run in polynomial time without relying on procedures such as max-min partitioning.


\begin{theorem} \label{3.1}
Given any instance $I=\langle{N,M,\fv}\rangle$ with $|N|=3$, a $\frac{2}{3}$-MMS allocation can be computed in polynomial time.
\end{theorem}

The case of three agents allows a complete characterization of $\frac{2}{3}$-MMS allocations using a reduction-based algorithm. The algorithm first removes any agent who can be satisfied by a single large item (value at least $\frac{2}{3}$). If no such item exists, we partition items into medium and small ones and search for a subset (called a {\em nice} bundle) that satisfies one agent while ensuring that the residual instance admits a $\frac{2}{3}$-MMS allocation for the remaining two agents. The main challenge lies in handling instances with exactly five medium items. In this case, we exploit structural properties: if two agents share a divisible item, we apply a cut-and-choose routine; if one agent has four or more indivisible items, a tailored reduction yields a nice bundle. These reductions are exhaustive and polynomial-time computable, which leads to a constructive proof of Theorem \ref{3.1}.

\begin{theorem} \label{4agent}
Given any instance $I=\langle{N,M,\fv}\rangle$ with $|N|=4$, a $\frac{2}{3}$-MMS allocation can be computed in polynomial time.
\end{theorem}

The four-agent case builds upon the ideas developed for three agents, but introduces new challenges that require more delicate treatment. Our algorithm follows a similar reduction strategy by first allocating any large item (value at least $\frac{2}{3}$) to a suitable agent. When no such item exists, we classify items as medium or small and identify reducible structures: either a \emph{nice bundle} that can be allocated to one or two agents without compromising fairness, or a \emph{temporary bundle} whose allocation may later be revised. The critical difficulty arises when a temporary bundle is removed and the reduced instance fails to meet the preconditions for applying the three-agent algorithm. To address this, we develop a \emph{check-and-withdraw} mechanism: whenever a temporary bundle threatens to reduce an agent's MMS below the threshold, we revert the reduction and identify an alternative one that guarantees progress. This process relies on a careful analysis of the structure of medium items, including their divisibility relations and agent preferences. By systematically enumerating all cases based on the number and type of medium items, we show that either a valid $\frac{2}{3}$-MMS allocation exists or a safe reduction can be performed. Each step is computable in polynomial time. This yields a constructive proof for Theorem \ref{4agent}.

\section{Conclusion}

We study MMS allocations under subjective divisibility, where agents may have different perceptions on whether an item is divisible. 
Our results are threefold. First, we show that even with unary valuations, the tight MMS approximation is $\frac{2}{3}$. Second, we design a new algorithm that guarantees a $\frac{5}{9}$-approximate MMS allocation in the general setting with both subjective divisibility and heterogeneous valuations, which improves the previous $\frac{1}{2}$ bound. Third, we give polynomial-time algorithms that compute $\frac{2}{3}$-MMS allocations for three and four agents.

There are many problems left open in this paper. 
The main open problem is to determine the tight approximation ratio for five or more agents. Our results for three and four agents hint strongly that the answer might be $\frac23$. Our techniques may shed light on this problem. 
Although we have shown that the tight approximation for unary valuations is $\frac23$, the worst-case scenario arises when there are at most $2n$ items.
For cases with more than 
$2n$ items, a more refined analysis can yield an improved approximation ratio, and we left this for subsequent research.

In addition to the open problems mentioned above, it is also worthwhile to explore more advanced settings. 
These include considering non-additive valuations (such as submodular or fractionally subadditive valuations), handling asymmetric agents (where agents have different entitlements to items), and examining the analogous problem of chore allocation (where agents prefer to receive as few items as possible).


\bibliographystyle{ACM-Reference-Format}
\bibliography{ref}


\newpage
\appendix

\section*{Appendix}











\section{Tight Approximation for Three and Four Agents} 
\label{ape:3agent}

In this section, we prove Theorems \ref{3.1} and \ref{4agent}.




\subsection{Additional Notations}

We begin by normalizing the instances, which ensures that our algorithms run in polynomial time, and introduce some preliminary techniques to preprocess the instances.  

\begin{definition}[Normalization]\label{2.1} 
Given any instance $I=\langle{N,M,\fv}\rangle$, we normalize the valuation of each agent $a_i\in N$ according to the following procedures. 
\begin{itemize}
    \item Initialize $M'=M$ and $n'=n$. \item Repeat the following steps until none occurs: 
    \begin{itemize}
        \item If there is one indivisible item $g\in N$ such that $v_i(g)\geq \frac{v_i(M')}{n'}$, $M'\leftarrow M'\setminus \{g\}$ and $n'\leftarrow n'-1$.
    \item If the $n'$-th and $(n'+1)$-th largest indivisible items $g_j$ and $g_{j'}$ in $M'$ satisfies $v_i(g_j)+v_i(g_{j'})\geq \frac{v_i(M')}{n'}$, $M'\leftarrow M'\setminus \{g_j,g_{j'}\}$ and $n' \leftarrow n'-1$.
    \end{itemize}
    \item Rescale the valuation such that $v_i(M')=n'$.
\end{itemize}


\end{definition} 


\begin{observation}
         For any normalized instance $\langle{N,M,\fv}\rangle$ and agent $a_i\in N$, $\MMS_i\leq 1$ and $v_i(M) \geq n$.
\end{observation}


\begin{algorithm}[h]
\caption{Allocating large items}
\label{alg:large:reduction}
\KwIn{Normalized instance $\langle{N, M, \fv}\rangle$.}
\KwOut{A partial allocation $(A_i)_{a_i \in N'}$ for $N'\subseteq N$.}
Initialize $A_i \leftarrow \emptyset$ for all $a_i\in N$, and $N'=\emptyset$.\\
\While{$|N|\geq 2$ and there exists one item $g\in M$ whose value is no less than $\frac{2}{3}$ for some agent in $N$}{
\eIf{item $g$ is indivisible for all agents $a_i$ with $v_i(g)\geq \frac{2}{3}$}
{
$A_i\leftarrow\{g\}$, $N' \leftarrow N'\cup \{a_i\}$, $M \leftarrow M\setminus \{g\}$ and $N \leftarrow N\setminus \{a_i\}$.\\
}{
$i\leftarrow \arg\max \{v_i(g) \,| \, g \text{ is divisible for } a_i\} $.\\
 $A_{i} \leftarrow \{\len_{i}(g, \frac{2}{3}) \cdot g\}$, $N' \leftarrow N'\cup \{a_i\}$, $M \leftarrow M\setminus \{\len_{i}(g, \frac{2}{3}) \cdot g\}$ and $N \leftarrow N\setminus \{a_i\}$.\\ 
}}
\If{$|N|= 1$}{
For agent $a_i\in N$, let $A_i\leftarrow M$.
}
\textbf{return} $N'$ and $(A_i)_{a_i\in N'}$.
\end{algorithm}


In the following, we only consider normalized instances.
We say an item $g$ is large to agent $a_i$ if $v_i(g)\geq \frac{2}{3}$, medium if $\frac{2}{3}>v_i(g)> \frac{1}{3}$, and small if $\frac{1}{3}\geq v_i(g)$.

\paragraph{Large Item Reduction}
Similar to Step 1 in Section \ref{sec:main}, we first process single large items; see Algorithm \ref{alg:large:reduction}.
Similar techniques have been used in \cite{DBLP:journals/aamas/BeiLLW21,DBLP:journals/corr/abs-2310-00976},
but we particularly take running time into consideration and our algorithm runs in polynomial time.
Consider the item that is large to some agents.
If the item is indivisible for all these agents, allocating this item alone to anyone does not decrease their MMS in the reduced instance. 
If it is divisible for some of them, we give a minimum fraction of $g$ to the one with maximum $v_i(g)$ among those seeing $g$ as divisible.
Then for any remaining agent $a_i$, $v_i(M\setminus \bigcup_{a_j\in N'}A_j)\geq |N\setminus N'|$. 
That is, allocating an indivisible or a minimum fraction of a divisible item does not decrease an agent's MMS or her average value in the reduced instance.  

\subsection{Two Agents}\label{sub:2agent}

We next introduce the notion of {\em nice bundle} in order to reduce an instance to a new one that guarantees a $\frac{2}{3}$-MMS allocation. 


\begin{definition}[Nice Bundle]\label{2.4} 
Given any normalized instance $\langle{N,M,\fv}\rangle$ with $|N|\geq 3$,  we call $M' \subseteq M$ a {\em nice bundle} if there exist two agents $\{a_1,a_2\}$ such that
\begin{itemize}
    \item there exists a $\frac{2}{3}$-MMS allocation among $N'= N\setminus\{a_1,a_2\}$ and $M'$.
    \item $v_{1}(M \setminus M') \geq 2$ and $ v_{2}(M \setminus M') \geq \frac{4}{3}$. 
\end{itemize}

\end{definition}

We call $N'$ the \textit{corresponding agent(s)} of $M'$.
We prove that after removing $N'$ and $M'$ from the instance, we can find a $\frac{2}{3}$-MMS allocation among the two agents $N\setminus N'$ and items $M\setminus M'$.


\begin{lemma}\label{2agent}
    Given a reduced normalized instance $\langle{N=\{a_1,a_2\}, M, \fv }\rangle$ where items in $M$ are not large for the two agents, with $v_1(M) \geq 2$ and $ v_{2}(M) \geq \frac{4}{3}$, Algorithm \ref{alg:2agent} returns a $\frac{2}{3}$-MMS allocation in polynomial time.
    



\end{lemma}

\begin{proof}
    Since every item is medium or small for agent $a_1$, we have $\frac{2}{3} \le v_1(B) < \frac{4}{3}$ and $v_1(B') > v_1(M) -  \frac{4}{3} \ge 2 - \frac{4}{3}=\frac{2}{3}$.
    Thus, no matter which bundle is left to her, agent $a_1$ is satisfied regarding $\frac{2}{3}\cdot\MMS_1$. 
    For agent $a_2$, since she selects the better one between $B$ and $B'$, her value is at least $\frac{1}{2}\cdot v_2(M)\geq \frac{1}{2} \cdot \frac{4}{3} = \frac{2}{3}$. So $(A_1, A_2)$ yields a $\frac{2}{3}$-MMS allocation.
\end{proof}

By Lemma \ref{2agent}, if we are able to find a nice bundle when the original instance contains more than two agents, then we directly obtain a $\frac{2}{3}$-MMS allocation.
Hence, we have the following corollary, which was proved in \cite{DBLP:journals/corr/abs-2310-00976}. The correctness of Algorithm \ref{alg:2agent} for two agents is established from Lemma \ref{2agent}.
We state it here with a polynomial implementation for completeness. 

\begin{lemma} \label{2a}
    Given any instance $I=\langle{N,M,\fv}\rangle$ with $|N|=2$, a $\frac{2}{3}$-MMS allocation can be computed in polynomial time.
\end{lemma}

\begin{algorithm}[t]
\caption{Computing a $\frac{2}{3}$-MMS allocation for two agents}\label{alg:2agent}
\KwIn{A reduced normalized instance $\langle{N=\{a_1,a_2\}, M, \fv}\rangle$ with no large items, $v_1(M) \geq 2$, $ v_{2}(M) \geq \frac{4}{3}$.}
\KwOut{ A $\frac{2}{3}$-MMS allocation.}
 Initialize $A_i \leftarrow \emptyset$, for $a_i\in N$.\\


Repeatedly adding arbitrary entire items to bag $B$ (empty initially) until $v_{1}(B) \in [\frac{2}{3}, \frac{4}{3})$.\\
 Give agent $a_2$ her preferred bundle between $B$ and $M \setminus B$ as $A_2$, and agent $a_{1}$ gets the remaining as $A_{1}$.\\ 

{\textbf{return} $(A_1, A_2)$}
\end{algorithm}

\subsection{Three Agents}\label{sec:3agents}

\begin{algorithm}[tb] 
\caption{Computing a $\frac{2}{3}$-MMS allocation for three agents}\label{alg1}
\KwIn{A normalized instance $\langle{N=\{a_1,a_2,a_3\},M,\fv }\rangle$.}
\KwOut{A $\frac{2}{3}$-MMS allocation for three agents.\par}
Initialize $A_1,A_2,A_3 \leftarrow \emptyset$.\\
\eIf{there exists an item $g\in M$ is large for some agent in $N$}{
 $(N', (A_i)_{a_i \in N'}) \leftarrow \text{Algorithm \ref{alg:large:reduction}}  (N, M, \fv)$.  \\ \par
 $N\leftarrow N\setminus N'$, $M\leftarrow M \setminus \bigcup_{a_i\in N'}A_i$. }
{ Divide $M$ into $G_1$ and $G_2$ by definition.\\
\uIf{$\exists a_i \in N, g \in G_1, v_i(G_2) \geq \frac{2}{3}-v_i(g)$ \label{alg3:case1begin}}{
 $B \leftarrow \{g\}$. Keep adding items from $G_2$ into $B$ until $\exists a_j \in N, v_{j}(B) \ge \frac{2}{3}$. \label{alg3:case1bagfill}\\
 $A_{j} \leftarrow B, N \leftarrow N \setminus {\{a_j\}}, M \leftarrow M \setminus A_j$. \label{alg3:case1end}\\}
\uElseIf{$|G_1| \geq 6$ \label{alg3:case2begin}}{
$(N',(A_i)_{a_i\in N'})\leftarrow$\text{Algorithm \ref{alg:2n}}($N, G_1$).\\
$N\leftarrow N\setminus N'$, $M\leftarrow M \setminus \bigcup_{a_i\in N'}A_i$.\\

\uIf{$|N|=1$}{
Allocate all the items in $M$ to the agent $a_j$ in $N$ as $A_j\leftarrow M$.
}
\ElseIf{$|N|=0$}{
Allocate all the items in $M$ to one agent $a_k$ in $N'$ arbitrarily as $A_k\leftarrow A_k\cup M$.\label{alg3:case2}\\
}}
     
 \uElseIf{
 there exist two agents $a_j, a_k$ who have the same divisible item $g_d\in G_1$
 \label{alg3:case3begin}}
{ 
Take two other arbitrary items $g_1, g_2$ from $G_1 \setminus \{g_d\}$.
 Find $A_j$, $A_k$ among $\{g_1,g_d,g_2\}$ by \textit{cut-and-choose} strategy (Definition \ref{def:cac}).\\
 Denote the agent in $N\setminus \{a_j,a_k\}$ by $a_l$. $A_l \leftarrow M\setminus\{g_1,g_d,g_2\}$, $N\leftarrow\emptyset$. \label{alg3:case3end}\\
 }
\Else{\label{alg3:case4begin}
Find the agent $a_j$ who considers more than 3 items in $G_1$ indivisible.\\
$B \leftarrow \max_{g \in G_1}{v_j(g)} \cup G_2$. $A_j \leftarrow B, N \leftarrow N \setminus {\{a_j\}}, M \leftarrow M \setminus B$. \label{alg3:case4end}\\ }

}
 \If{$|N|=2$}{
  $(A_i)_{a_i \in N} \leftarrow$ \text{Algorithm \ref{alg:2agent}}($N, M$).\\
}
   {\textbf{return} $(A_1, A_2, A_3)$}
\end{algorithm}


        

        

In this subsection, we provide Algorithm \ref{alg1} for 3 agents.
The algorithm also works for the reduced instance, if the maximin share of the remaining agents doesn't decrease compared to the original instance.


If there exists any large item(s), Algorithm \ref{alg:large:reduction} either reduces the instance to an instance with 2 agents or computes a $\frac{2}{3}$-MMS allocation.
For the reduced instance with two agents, we directly use the method mentioned in subsection \ref{sub:2agent} to find a $\frac{2}{3}$-MMS allocation.
Hence, in the remaining part of the analysis, we only consider the instance with 3 agents that contains no item which is large for any agent in the instance.

The high-level idea of our algorithm is to focus on the set of medium items, namely $G_1$. 
We again define $G_1$ and also $G_2$ here:
$$G_1:=\{g \in M \, | \, \exists a_i \in N, v_i(g) > \frac{1}{3}  \}; $$
$$G_2:=\{g \in M \, | \, \forall a_i \in N, v_i(g) \leq \frac{1}{3}  \}.$$

To prove the correctness of Theorem \ref{3.1}, we only need to show there exists either a $\frac{2}{3}$-MMS allocation or a nice bundle in the four following cases:
\begin{itemize}
    \item Case 1: There exist $a_i \in N$ and $ g \in G_1$ such that $ v_i(G_2) \geq \frac{2}{3} - v_i(g)$. {(Lines \ref{alg3:case1begin}-\ref{alg3:case1end}, Lemma \ref{3.2})}

    \item Case 2: For each $a_i \in N$ and $g \in G_1$, $ v_i(G_2) < \frac{2}{3} - v_i(g);|G_1| \geq 6$. {(Lines \ref{alg3:case2begin}-\ref{alg3:case2}, Lemma \ref{3.4})}

    \item Case 3: For each $a_i \in N$ and $g \in G_1$, $ v_i(G_2) < \frac{2}{3}- v_i(g);|G_1|=5$ and there exist two agents $a_j, a_k$ who share the same divisible item $g_d \in G_1$ {(Lines \ref{alg3:case3begin}-\ref{alg3:case3end}, Lemma \ref{3.6})}

    \item Case 4: For each $a_i \in N$ and $g \in G_1$, $ v_i(G_2) < \frac{2}{3} - v_i(g);|G_1|=5$ and there exists an agent $a_j$ who has at least four indivisible items in $G_1$. {(Lines \ref{alg3:case4begin}-\ref{alg3:case4end}, Lemma \ref{3.7})}
\end{itemize}

It is noted that whenever there is no item in $G_1$ divisible to at least two agents, there exists an agent who has at least four indivisible items in $G_1$.
We will further know that these four cases contain all possible situations, which we postpone to Lemma \ref{3.3}.





For Case 1, when $\exists a_i \in N, g \in G_1, v_i(G_2) \geq \frac{2}{3} - v_i(g)$, our algorithm first initializes bag $B$ as $\{g\}$.
Then we keep adding items from $G_2$ to $B$ until there is one agent $a_j$ in $N$ who values $B$ at least $\frac{2}{3}$.
After that, we allocate $B$ to agent $a_j$ as $A_j$ and remove them from the instance. 
Hence, $v_j(A_j)\geq \frac{2}{3}\geq \frac{2}{3}\MMS_j$.

\begin{lemma} \label{3.2}
    If there exist $a_i \in N$ and $ g \in G_1$ such that $ v_i(G_2) \geq \frac{2}{3} - v_i(g)$,  then the set of items $B$ selected in Line \ref{alg3:case1begin}-\ref{alg3:case1end} is a nice bundle.
\end{lemma}

    \begin{proof}
        Given that every agent in $N$ values any item in $G_1$
        less than $\frac{2}{3}$, and any item in $G_2$ no more than $\frac{1}{3}$, no agent in $N\setminus {a_i}$ values bundle $B$ at greater than $1$.
        Therefore, $\forall a_k \in N \setminus a_{j}, v_{k}(M \setminus B) \geq 3-1=2$. 
        According to Definition \ref{2.4}, it implies that bundle $B$ is a nice bundle with $N'=\{a_j\}$.
    \end{proof}

    
Before discussing next three cases, we first bound the number of items in $G_1$ if $\forall a_i \in N, g \in G_1, v_i(G_2) < \frac{2}{3}-v_i(g)$, which also means there are no other possible situations beyond Cases 1-4.

\begin{lemma} \label{3.3}
     If $v_i(G_2) < \frac{2}{3}-v_i(g)$ for each $a_i \in N$ and $g \in G_1$, $|G_1| \geq 5$.
\end{lemma}

     \begin{proof}
         It easily follows that $\forall a_i \in N, \forall g \in G_1, v_i(G_1) \geq \frac{7}{3}+v_i(g)$. 
         By definition, each item in $G_1$ is either medium or small to every agent $a_i$, we have $|G_1 \setminus \{g\}| \geq 4$, which implies $|G_1| \geq 5$.
     \end{proof}

For Case 2, if $\forall a_i \in N, \forall g \in G_1, v_i(G_2) < \frac{2}{3} -v_i(g)$ and $|G_1|\geq 6$, we use Algorithm \ref{alg:2n} below to find a partial $\frac{2}{3}$-MMS allocation of the instance and reduce the instance by removing the partial allocation.
For the reduced instance returned by Algorithm \ref{alg:2n}, if the reduced instance contains 2 agents, the set of items $\bigcup_{a_i\in N'}A_i$ is a nice bundle and its corresponding agent is also removed.
If the reduced instance contains 1 agent, we just allocate all the remaining items to her.
If the reduced instance contains no agent, it means that every agent has already got a bundle with value larger than $\frac{2}{3}$.
In either way a $\frac{2}{3}$-MMS allocation is finally achieved.

\begin{lemma}\label{3.4}
    If $v_i(G_2) < \frac{2}{3} -v_i(g)$ for each $a_i \in N$ and $g \in G_1$, and $|G_1|\geq 6$, we can find a set of agents $N'$ and their allocation $(A_i)_{a_i\in N'}$ by Algorithm \ref{alg:2n} with the following two properties hold:
     \begin{enumerate}[label=(\roman*)]
        \item Every agent $a_i$ in $N'$ are allocated two items $\{g_1,g_2\}$ from $G_1$ that $v_i(g_1)>\frac{1}{3}$ and $v_i(g_2)>\frac{1}{3}$; \label{g1:2}

        \item Every agent $a_j$ in $N\setminus N'$ values $\bigcup_{a_i\in N'}A_i$ no more than $\frac{2}{3}|N'|$. \label{g1:3}
    \end{enumerate}
\end{lemma}

\begin{proof}
In Algorithm \ref{alg:2n}, we first find the set of all the combinations of agents $\Phi$ in $N$.
Let $N'=N$.
For any set of agents $B$, we find the set of items $S=\{g\in G_1 \,|\, \exists a_i\in N,v_i(g)>\frac{1}{3}\}$.
If $|S|<2|B|$, we remove the agents $B$ and items $S$ from $N'$ and $G_1$.
After the removal, we have $|G_1|\geq 2|N'|$ and for every item $g$ in $G_1$, there exists an agent $a_i$ in $N'$ such that $v_i(g)>\frac{1}{3}$ and every agent $a_j$ in $N\setminus N'$ thinks that $v_j(g)\leq \frac{1}{3}$.
In addition, there is no subset $B'$ of $N'$ and the set of items $S'=\{g\in G_1|\exists a_i\in B',v_i(g)>\frac{1}{3}\}$ such that $|S'|<2 |B'|$ even $|B'|=1$.

For the remaining agents in $N'$ and items in $G_1$, we can allocate each agent $a_i\in N'$ two items $g_1,g_2$ in $G_1$ such that 
$v_i(g_1)>\frac{1}{3}$ and $v_i(g_2)>\frac{1}{3}$.
For every agent $a_j\in N\setminus N'$, $v_j(\bigcup_{a_i\in N'} A_i)\leq \frac{2}{3}|N'|$.
Hence, the lemma is proved.
\end{proof}

\begin{algorithm}[tb] 
 \caption{Reduction via $G_1$} \label{alg:2n}
\KwIn{Agents $N$, items $G_1=\{g_1, g_2,\ldots\} \,\, (|G_1| \geq 2|N|)$.}
\KwOut{A $\frac{2}{3}$-MMS allocation for agents $N'$.}

Find the set of all the combinations $\Phi$ of agents in $N$.\\
Initialize $N'\leftarrow N$. \\
\For{$B\in \Phi$}{
Find the set of items $S=\{g\in G_1 \,| \,\exists a_i\in B,v_i(g)>\frac{1}{3}\}$.\\
\If{$|S|<2|B|$}{
$N'\leftarrow N'\setminus B$, $G_1\leftarrow G_1\setminus S$.
}
}
Find an allocation $(A_i)_{a_i\in N'}$ such that every agent $a_i$ is allocated two items $g_1$ and $g_2$ in $G_1$ such that $v_i(g_1)>\frac{1}{3}$ and $v_(g_2)>\frac{1}{3}$.



\textbf{return} $N'$ and $(A_i)_{a_i \in N'}$
\end{algorithm}


 
In the following part, we focus on two cases when $|G_1|=5$: there exists an agent that has at least four indivisible goods in $G_1$, or each agent has at most three indivisible goods in $G_1$. 
Before that, we want to find some properties for Cases 3 and 4.

\begin{observation} \label{3.5}
    When $|G_1|=5$ and $v_i(G_2 \cup \{g\})<\frac{2}{3}$ for every item $g\in G_1$ and every agent $a_i\in N$, the following three properties hold:
    \begin{enumerate}[label=(\roman*)]
        \item $\forall \{g_1, g_2, g_3\} \subseteq G_1, \forall a_i \in N, v_i(\{g_1,g_2,g_3\}) \geq \frac{5}{3} $; \label{3.5-1}
        \item $\forall \{g_1, g_2\} \subseteq G_1, \forall a_i \in N, v_i(\{g_1,g_2\}) \geq 1$; \label{3.5-2}
        \item $\forall g \in G_1, \forall a_i \in N, v_i(g) \geq \frac{1}{3} $. \label{3.5-3}
    \end{enumerate} 
\end{observation}

    \begin{proof}
        Fix any agent $a_i$, any three items $\{g_1,g_2,g_3\}\subseteq G_1$ and any other item $g$ in $G_1\setminus\{g_1,g_2,g_3\}$. 
        Since $v_i(G_1 \setminus \{g\}) >\frac{7}{3}$ and $v_i(g)<\frac{2}{3}$, we have 
        $v_i(\{g_1,g_2,g_3\}) > \frac{5}{3}$,
        then property \ref{3.5-1} holds.
        
        Based on the property \ref{3.5-1}, $\forall a_i\in N$ and $\forall \{g_1,g_2\}\subseteq G_1$, we have $v_i(g_1,g_2)\geq \frac{5}{3}-\frac{2}{3}=1$.
        Also, $\forall a_i\in N$ and $\forall g\in G_1$, we have $v_i(g)\geq 1-\frac{2}{3}=\frac{1}{3}$. Thus, property \ref{3.5-2} and \ref{3.5-3} are implied.
    \end{proof}

    Based on Observation \ref{3.5}, every item in $G_1$ is medium to each agent in Cases 3 and 4.

 For Case 3 where $v_i(G_2) < \frac{2}{3} - v_i(g)$ for each $a_i \in N$ and $g \in G_1$, $|G_1|=5$ and there exist two agents $a_j, a_k$ who share the same divisible item $g_d \in G_1$, we arbitrarily choose two items from $G_1\setminus \{g_d\}$, denoted by $g_1$ and $g_2$.
 Via the modified classical cut-and-choose strategy, $\{g_1,g_d,g_2\}$ are divided into disjoint bundles $A_j$ and $A_k$ to allocate to $a_j$ and $a_k$, with all remaining items allocated to another agent $a_l$.
 
\begin{definition}[Cut-and-Choose Strategy]\label{def:cac}
    Given a set of three items $\{g_1,g_d,g_2\}$ and two agents $\{a_1,a_2\}$  such that $g_d$ is divisible for them and $v_i(g_p)\leq \frac{1}{2}v_i(\{g_1,g_d,g_2\})$ for $a_i\in \{a_1,a_2\}$, $g_p\in \{g_1,g_2\}$.

    The \textit{cut-and-choose} strategy picks one agent arbitrarily, say $a_1$, lets her cut the item $g_d$ once and divides the set of three items into two bundles of equal value for $a_1$.
    Then we assign agent $a_2$ her preferred bundle, denoted by $A_2$, with the other bundle $A_1$ allocated to $a_1$.
    It is straightforward to see that this strategy ensures $a_1$ and $a_2$ get a bundle at least $\frac{1}{2}v_i(\{g_1,g_d,g_2\}),i=1,2$.
\end{definition}

\begin{lemma} \label{3.6}
    When $ v_i(G_2) < \frac{2}{3} - v_i(g)$ for each $a_i \in N$ and $ g \in G_1$, $|G_1|=5$ and every agent has at most three indivisible items in $G_1$, a $\frac{2}{3}$-MMS allocation always exists.
\end{lemma}

    \begin{proof}
        Assume w.l.o.g. that agent $a_1,a_2$ regard $g_d$ as divisible item. 
        We then pick two arbitrary items in $G_1\setminus \{g_d\}$, denoted by $g_1,g_2$. 
        As seen in Observation \ref{3.5},
        $v_{i}(\{g_1,g_d,g_2\}) > \frac{5}{3}$ for $a_i\in\{a_1,a_2\}$.
         Recall that both agents values each good at most $\frac{2}{3}$, $v_i(g_p)<\frac{\{g_1,g_d,g_2\}}{2}$ for $a_i\in \{a_1,a_2\}$, $g_p\in \{g_1,g_2\}$. 


        Let $a_1$ cut and $a_2$ choose.
        It is clear that both $a_{1}$ and $a_{2}$ get a bundle with value at least $\frac{1}{2}*\frac{5}{3}>\frac{2}{3}$.
        For the remaining agent $a_3$, we have $v_3(\{g_1,g_d,g_2\})\leq 3*\frac{2}{3}=2$. 
        Therefore,  $v_3(M\setminus \{g_1,g_d,g_2\})\geq 3-2=1$ and we get a $\frac{2}{3}$-MMS allocation.
    \end{proof}

For Case 4 where $v_i(G_2) < \frac{2}{3} - v_i(g)$ for each $a_i \in N$ and $g \in G_1$, $|G_1|=5$ and there exists an agent $a_j$ who has at least four indivisible items in $G_1$,
the algorithm allocates $a_j$ with a set of items $\max_{g \in G_1}{v_i(g)} \cup G_2$ as $A_j$ and removes them from the instance.
Different from Lemma \ref{3.6}, the proof of Lemma \ref{3.7} considers the maximin share of agents.

\begin{lemma} \label{3.7}
   When $ v_i(G_2) < \frac{2}{3}- v_i(g)$ for each $a_i \in N$ and $ g \in G_1$, $|G_1|=5$ and there exists an agent $a_j$ who has at least four indivisible items in $G_1$, then $\max_{g \in G_1}{v_j(g)} \cup G_2$ is a nice bundle and $a_j$ is the corresponding agent.
\end{lemma}
   
\begin{proof}
    Without loss of generality, let $g_k=\max_{g \in G_1}{v_j(g)}$.
    We first prove that after allocating the bundle $g_k \cup G_2=A_j$ to agent $a_j$, other two agents have enough value for the four items in $G_1 \setminus \{g_k\}$. 
    Since $v_i(g_k \cup G_2)<\frac{2}{3}$ for every agent $i$, the value of the remaining items is at least $\frac{7}{3}>2$ for both remaining agents. 
    Therefore, to show $A_j$ is a nice bundle and its corresponding agent is $a_j$, we are left to show that bundle $A_j$ is large enough for agent $a_j$, i.e., $v_j(A_j) \geq \frac{2}{3}$.

    Since agent $a_j$ considers at least $4$ items in $G_1$ indivisible, for agent $a_j$, denote the third-largest indivisible item by $g_p$ and the fourth-largest indivisible item by $g_q$.
    If $v_j(g_p)+v_j(g_q)>1$, by the definition of normalized instance, $v_j(M\setminus \{g_p, g_q\}) \geq 2$.
    Otherwise, if $v_j(g_p)+v_j(g_q)\leq 1$, we still have $v_j(M\setminus \{g_p, g_q\}) \geq 2$.
    Since the value of the two items in $G_1\setminus \{g_k,g_p,g_q\}$ is at most $\frac{4}{3}$, we have $v_j(A_j)\geq 2-\frac{4}{3}=\frac{2}{3}$.

    Therefore, $A_j$ is a nice bundle and $a_j$ is the corresponding agent.
   \end{proof}

Based on the above discussion from Lemma \ref{3.1} to \ref{3.7}, 
Algorithm \ref{alg1} computes a $\frac{2}{3}$-MMS allocation for any instance with 3 agents via polynomial time, which completes the proof of Theorem \ref{3.1}. 

It is worth noting that only in Case 4, the property of maximin share is used.
Hence, if we want to find a $\frac{2}{3}$-MMS allocation in a reduced instance (the original instance contains at least 4 agents), we must ensure the MMS of each agent in the reduced instance does not decrease compared to that in the original instance in Case 4.

\subsection{Four Agents}\label{ape:4agent}

\begin{algorithm}[tb] 
\caption{Computing a $\frac{2}{3}$-MMS allocation for four agents}\label{alg2}
\KwIn{A normalized instance $\langle N=\{a_1,a_2,a_3,a_4\}, M,\fv \rangle$.}
\KwOut{A $\frac{2}{3}$-MMS allocation for four agents.\par}
Initialize $A_1,A_2,A_3,A_4 \leftarrow \emptyset$. \\
\eIf{there exists an item $g\in M$ is large for some agent in $N$ \label{alg3:case0begin}}{
$(N', (A_i)_{a_i \in N'}) \leftarrow \text{Algorithm \ref{alg:large:reduction}}(\langle{N, M, v}\rangle)$.\\
$N\leftarrow N\setminus N'$, $M\leftarrow M\setminus \bigcup_{a_i\in N'}A_i$. \\}
{
Divide $M$ into $G_1$ and $G_2$ by definition. \\

\uIf{$\exists a_i \in N, g \in G_1, v_i(G_2) \geq \frac{2}{3} - v_i(g)$ \label{alg4:case1start}}{
$B \leftarrow \{g\}$. Keep adding the items from $G_2$ to $B$ until $\exists a_j \in N, v_{j}(B) >\frac{2}{3}$. \\
$A_{j} \leftarrow B, N \leftarrow N \setminus {\{a_j\}}, M \leftarrow M \setminus B$. \label{alg4:case1end}
 }

\uElseIf{$|G_1| \geq 8$ \label{alg4:case2start}}{
    
$(N',(A_i)_{a_i\in N'}) \leftarrow$ Algorithm \ref{alg:2n}$(N,M)$. \label{alg4:case2}\\
$N\leftarrow N\setminus N'$, $M\leftarrow M\setminus \bigcup_{a_i\in N'}A_i$.\\

\uIf{$|N|=1$}{
Allocate all the items in $M$ to agent $a_j$ in $N$ as $A_j\leftarrow M$.
}
\ElseIf{$|N|=0$}{
Allocate all the items in $M$ to one agent $a_k$ in $N'$ arbitrarily as $A_k\leftarrow A_k\cup M$.\label{alg4:case2}\\
}
}
  \uElseIf { there exist two agents $a_j$ and $a_k$ in $N$ who have the same divisible item $g_d \in G_1$
  \label{alg4:case3start}}{
Take two arbitrary items $g_1, g_2$ from $G_1 \setminus \{g_d\}$.
Find $A_j$ and $A_k$ among $\{g_d,g_1,g_2\}$ by \textit{cut-and-choose} strategy (Definition \ref{def:cac}).\\
$M\leftarrow M\setminus \{g_d,g_1,g_2\}$, $N\leftarrow N\setminus\{a_j,a_k\}$. \label{alg4:case3end}\\
        }
     \Else{ \label{alg4:case4start}
     Find one agent $a_j\in N$ who has at least 6 indivisible items in $G_1$.\\
     $B \leftarrow \max_{g \in G_1}{v_j(g)} \cup G_2$.\\
 $A_j \leftarrow B, N \leftarrow N \setminus {a_j}, M \leftarrow M \setminus A_j$. \label{alg4:case4end}\\
 }
}
\uIf{$|N|=3$}{
 Update the instance and the allocation by \textit{check-and-withdraw} strategy. \label{alg4:withdraw}\\
$(A_i)_{a_i \in N} \leftarrow$ Algorithm \ref{alg1}($N, M$).\\ 
} 
\ElseIf{$|N|=2$}{      
$(A_i)_{a_i \in N} \leftarrow$ Algorithm \ref{alg:2agent}($N, M$).
}
{\textbf{return} $(A_1, A_2, A_3, A_4)$}
\end{algorithm}

In this section, we introduce the algorithm for four agents, whose structure is similar to the algorithm for three agents.

We still try to find a reduction to reduce the number of agents.
However, when the instance is reduced to one with three agents,
if the maximin share of some remaining agent decreases, Algorithm \ref{alg1} may not find a $\frac{2}{3}$-MMS allocation. 
Our core technique is a \textit{check-and-withdraw strategy} for some special cases to prevent it happening.

Similar to the definition of nice bundle, we define temporary bundle as follow:

\begin{definition}[Temporary Bundle] \label{temporary}
Given any normalized instance $I=\langle{N=\{a_1,a_2,a_3,a_4\},M}\rangle$, we call $M'\subseteq M$ a \textit{temporary bundle} if there exists one agent $N'=\{a_i\}$ such that
\begin{itemize}
    \item $v_i(M')\geq \frac{2}{3}\MMS_i$.
    \item $v_j(M\setminus M')\geq 3$ for any agent $a_j\in N\setminus \{a_i\}$.
\end{itemize}
We call $a_i$ the corresponding agent of $M'$.


\end{definition}


\subsubsection{Main Algorithm for Four Agents}
We use Algorithm \ref{alg2} to compute the $\frac{2}{3}$-MMS allocation for four agents.
At first, if there exists any large item(s) to some agent in $N$, we use Algorithm \ref{alg:large:reduction} to reduce the instance. It is noted that the item itself is a temporary bundle.




After that, we still divide the items $M$ into $G_1$ and $G_2$ defined in Subsection\ref{sec:3agents}.

Our intuitive idea is to prove there always exists either a nice bundle, a temporary bundle, or a $\frac{2}{3}$-MMS allocation in the following four cases:
\begin{itemize}
    \item  Case 1: There exist $a_i \in N$ and $ g \in G_1$ such that $ v_i(G_2) \geq \frac{2}{3} - v_i(g)$. {(Lines \ref{alg4:case1start}-\ref{alg4:case1end}, Lemma \ref{4-1})}

    \item  Case 2: For each $a_i \in N$ and $ g \in G_1, v_i(G_2) < \frac{2}{3}- v_i(g);|G_1| \geq 8$. {(Lines \ref{alg4:case2start}-\ref{alg4:case2}, Lemma \ref{4-2})}

    \item Case 3: For each $a_i \in N$ and $ g\in G_1, v_i(G_2) < \frac{2}{3}- v_i(g);|G_1|=7$ and there exist two agents $a_j$ and $a_k$ who share the same divisible item $g_d\in G_1$. {(Lines \ref{alg4:case3start}-\ref{alg4:case3end}, Lemma \ref{4-3})}

    \item Case 4: For each $a_i \in N$ and $g \in G_1, v_i(G_2) < \frac{2}{3} - v_i(g);|G_1|=7$ and there exists an agent who has at least six indivisible items in $G_1$. {(Lines \ref{alg4:case4start}-\ref{alg4:case4end}, Lemma \ref{4-4})}
    

\end{itemize}

It is noted that whenever $|G_1|=7$ and there exists no item in $G_1$ that is divisible for two agents, there exists one agent who has at least $6$ indivisible items.
The analysis of all four cases are highly analogous to those in the last section, so we won't give the complete proofs here.


\begin{lemma} [Case 1]\label{4-1}
    If there exist $a_i \in N$ and $g \in G_1$ such that $ v_i(G_2) \geq \frac{2}{3} - v_i(g)$, then the set of items $B$ selected in Line  \ref{alg4:case2start}-\ref{alg4:case2} is a temporary bundle.
\end{lemma}




\begin{lemma}\label{4-2-0}
     If $v_i(G_2) < \frac{2}{3}- v_i(g)$ for all $a_i \in N$ and $ g \in G_1$, then $|G_1|\geq 7$.
 \end{lemma}


\begin{lemma}[Case 2] \label{4-2}
    When $ v_i(G_2) < \frac{2}{3}- v_i(g)$ for each $a_i \in N$ and $ g \in G_1$ , and $|G_1| \geq 8$, we can find a set of agents $N'$ and their allocation $(A_i)_{a_i\in N'}$ by Algorithm \ref{alg:2n} with two below properties hold:
    \begin{enumerate}[label=(\roman*)]
        \item Every agent $a_i$ in $N'$ are allocated two items $\{g_1,g_2\}$ in $G_1$ that $v_i(g_1)>\frac{1}{3}$ and $v_i(g_2)>\frac{1}{3}$; 

        \item Every agent $a_j$ in $N\setminus N'$ values $\bigcup_{a_i\in N'}A_i$ no more than $\frac{2}{3}|N'|$.
    \end{enumerate}
\end{lemma}



\begin{lemma} \label{4-2-1}
    If $v_i(G_2 \cup \{g\}) < \frac{2}{3}$ for each $a_i\in N$ and $ g \in G_1$  and $|G_1|=7$, then $ v_i(\{g_1,g_2,g_3\}) \geq \frac{4}{3}$ for any $g_1, g_2, g_3 \in G_1$.
\end{lemma}



\begin{lemma} [Case 3]\label{4-3}
    When $v_i(G_2) < \frac{2}{3} - v_i(g)$ for each $a_i \in N$ and $ g \in G_1$, $|G_1|=7$ and every agent has at most five indivisible items in $G_1$, a $\frac{2}{3}$-MMS allocation always exists.
\end{lemma}



\begin{lemma} [Case 4]\label{4-4}
    When $v_i(G_2) < \frac{2}{3} - v_i(g)$ for each $a_i \in N$ and $ g \in G_1$, $|G_1|=7$ and there exists an agent $a_j$ who has at least six indivisible items in $G_1$,  $B = \max_{g \in G_1}{v_j(g)} \cup G_2$ is a temporary bundle and agent $a_j$ is its corresponding agent.
\end{lemma}

Up to now, we have proved that we can find either a nice bundle, a temporary bundle, or a $\frac{2}{3}$-MMS allocation in each case.
In the next subsection, we are left to show how to use check-and-withdraw strategy to modify the cases with a temporary bundle, which is the core technique of our algorithm for four agents.

\subsubsection{The Check-and-Withdraw Mechanism} \label{sub:2}
In this part, we firstly summarize the conditions when the algorithm for 3 agents no longer works after finding a temporary bundle, and then show that another reduction could be found in polynomial time, which will prevent this situation from happening. 

For ease of presentation, denote the remaining items after removing the temporary bundle by $M_r$ and the remaining agents by $N_r$.
After the removal, $M_r$ are further divided into two categories with respect to the remaining agent $N_r$: 
$$\hat{G}_1=\{g\in M_r \,| \,\frac{1}{3}<v_i(g), \exists a_i\in N_r\},$$  $$\hat{G}_2=\{g\in M_r\,|\, v_i(g)\leq \frac{1}{3}, \forall a_i\in N_r\}.$$
It is noted that Algorithm \ref{alg1} may not work after removing a temporary bundle only when the maximin share of the remaining agents is considered, which refers to Case 4 in Algorithm \ref{alg1}.

It is easy to observe that in Algorithm \ref{alg1}, Case 4 happens if all of the five conditions hold:
         \begin{itemize}
            \item (1) There exists no item $g$ in $M_r$ and agent $a_i\in N_r$ such that $v_i(g)\geq \frac{2}{3}$ (not entering Line \ref{alg3:case0begin}), \label{3a0}
            \item (2) $\forall a_i\in N_r, g\in \hat{G}_1, v_i(\hat{G}_2\cup \{g\})<\frac{2}{3},$ (not entering Line \ref{alg4:case1start}), \label{3a2}
            \item (3) $|N_r|=3$ and $|\hat{G}_1|=5$ (not entering Line \ref{alg3:case2begin}), \label{3a1}
            \item (4) Every item $g\in \hat{G}_1$ is valued over $\frac{1}{3}$ by any agent $a_i\in N_r$ (Lemma \ref{3.5})\label{3a3},
            \item (5) There exists no shared divisible item for two agents in $\hat{G}_1$. (not entering Line \ref{alg4:case3start}). \label{3a4}
        \end{itemize}

For simplicity, denote the remaining agents by $\{a_1,a_2,a_3\}$, the removed agent by $a_4$, the temporary bundle by $A_4$, and the items in $\hat{G}_1$ by $\{g_1,g_2,g_3,g_4,g_5\}$.
Apart from these five conditions, if there exists one agent $a_i$ who considers all the five items in $\hat{G}_1$ indivisible, after allocating two items in $G_1$ with smallest values of her to others, $a_i$ values the remaining items still no less than $2\cdot\MMS_i$, hence the algorithm with 3 agents can still work.
We allocate bundle $\{g_1,g_2\}$ to one agent in $N_r\setminus \{a_i\}$ arbitrarily and prove it is a nice bundle for the reduced instance.

\begin{lemma}
   If there is an agent $a_i$ who has 5 indivisible items in $\hat{G}_1$ and Conditions (1)-(5) are satisfied, without loss of generality assume that $a_i$'s two items with the smallest value in $\hat{G}_1$ are $g_1$ and $g_2$.
   Then $\{g_1,g_2\}$ is a nice bundle for the reduced instance.
\end{lemma}
\begin{proof}
For agent $a_i$, there are five indivisible items in $\hat{G}_1$.
According to the definition of normalization, for agent $a_i$, if $a_i$ values $\{g_1,g_2\}$ greater than 1, then the value of remaining items is $3$. 
If the value is smaller than 1, then the value of remaining items is still greater than $3$.
Since the value of the temporary bundle for agent $a_i$ is less than or equal to $1$, she values the remaining goods at least $2$.

For each agent $a_j$ in $N_r\setminus{\{a_i\}}$ who has not been allocated any item, the value of $A_4$ is less than $1$, and the value of the two items in $\hat{G}_1$ is at most $\frac{4}{3}$. 
Therefore, she values all the remaining items by at least $\frac{5}{3}$.
Hence, ${g_1, g_2}$ satisfy the definition of a nice bundle for the reduced instance and $a_j$ is the corresponding agent.
\end{proof}

Based on that, we add the sixth condition for further analysis:
\begin{itemize}
            \item (6) Each agent considers at least one item in $\hat{G}_1$ divisible.\label{3a5}
\end{itemize}

\begin{claim} \label{cla:rem1}
        If one remaining agent in $N_r$ values the temporary bundle $A_4$ by no more than $\frac{2}{3}$, Condition (1) or (2) cannot hold simultaneously.
    \end{claim}
    \begin{proof}
        Denote the agent who values the temporary bundle $A_4$ no more than $\frac{2}{3}$ by $a_i$.
        Suppose Condition (1) is satisfied,
        for agent $a_i$ and any item $g\in \hat{G}_1$, the value of $\hat{G}_1\setminus \{g\}$ is at most $4\cdot\frac{2}{3}=\frac{8}{3}$.
        Hence, 
        
        \[
        v_i(\hat{G}_2\cup \{g\})\geq v_i(\hat{G}_1\cup \hat{G}_2)-v_i(\hat{G}_1\setminus \{g\})\geq 4-\frac{2}{3}-\frac{8}{3}=\frac{2}{3},
        \]
        contradicting to Condition (2).
    \end{proof}
    
We now consider the situation when we reduce the instance with 4 agents by giving one single item $g$ to one agent who values it by no less than $\frac{2}{3}$.
For any agent $a_i$ in $N_r$, if $g$ is indivisible, the maximin share of her does not decrease.
If $g$ is divisible for agent $a_i$, $g$ is allocated partially and the value of the allocated fraction of $g$ for her is no larger than $\frac{2}{3}$.  
By Claim \ref{cla:rem1}, Condition (1) or (2) must be broken and Algorithm \ref{alg1} can still work. 
Hence we have an additional condition:
\begin{itemize}
            \item (7) At least two items are removed.\label{3a6}
\end{itemize}

\begin{lemma}\label{lem:fix}
    After the removal of $a_4$ and $A_4$ with Conditions (1)-(7) satisfied, we withdraw the agent $a_4$ and the temporary bundle $A_4$.
    A new reduction can be found in polynomial time such that the new reduced instance doesn't meet at least one of the seven conditions. 
\end{lemma}


\begin{proof}
    Depending on Condition (7), there are at least two items in $A_4$.
    Hence, here we consider four cases discussed in the Algorithm \ref{alg2}. 
    Now we withdraw the agent $a_4$ and the set of items $A_4$, trying to find a new reduction $N^*$ and $(A_{i})_{a_i\in N^*}$.
    Denote the new remaining agents as $N'_r=N\setminus N^*$ and the new remaining items as $M'_r=M\setminus \bigcup_{a_i\in N^*}A_i$.


    \textit{Case 1:} There exist $a_i \in N$ and $g \in G_1$ such that $v_i(G_2) \geq \frac{2}{3} - v_i(g)$.
    This case is the most intricate one. Algorithm \ref{alg2} only allocates one item in $G_1$ to agent $a_4$.
    Therefore, $|G_1|\geq |\hat{G_1}|+1=6$ and we discuss two subcases according to the cardinality of $G_1$.

    \textit{Subcase 1.1:} $|G_1|=6$.
    Denote $G_1=\{g_1,\ldots,g_6\}$.
     Case 1 is divided into the following 3 subcases.
     Under Condition (7), every agent in $N_r$ considers at least one item in $G_1$ divisible. 
    Based on Conditions (4) and (5), at least one agent considers there is only one divisible item in $\hat{G}_1$.
    For any agent $a_i\in N_r$ and any item $g\in G_1$, the value of $G_2\cup \{g\}$ satisfies $v_i(G_2\cup \{g\}) \geq 4-5\cdot \frac{2}{3}=\frac{2}{3}$. 

    \begin{itemize}

    \item \textit{Subsubcase 1.1.1:} There exists one agent in $N_r$ who considers there is only one divisible item in $G_1$.
    Let the agent and item be $a_1$ and $g_1$ respectively.
    We put item $g_1$ into the bag $B$ and keep adding the items in $G_2$ until some agent thinks the value of $B$ is no less than $\frac{2}{3}$.
    If there is an agent $a_i$ in $N\setminus \{a_1\}$ who values $B$ no less than $\frac{2}{3}$, allocate $B$ to this agent and remove them.
    After this reduction, agent $a_1$ in $N'_r$ thinks the items in $G_1$ are all indivisible, which breaks Condition (6).
    Otherwise, there are three agents who value $B$ less than $\frac{2}{3}$.
    We allocate $B$ to agent $a_1$ as $N_1$ and remove them to get a new instance with 3 agents.
    Based on Claim \ref{cla:rem1}, the new instance does not satisfy Conditions (2) and (3) at the same time.


    \item \textit{Subsubcase 1.1.2:} There are two agents who share a divisible item in $G_1$.
    Let the two agents and the item be $a_1,a_2$ and $g_1$, respectively.
    $g_1$ is firstly added into bag $B$ and then items in $G_2$ are added into $B$ one by one until some agent values $B$ by no less than $\frac{2}{3}$.
    If one agent in $N\setminus\{a_1,a_2\}$ values $B$ no less than $\frac{2}{3}$, allocate $B$ to the agent and remove them.
    After this reduction, the reduced instance violates Condition (6).
    Otherwise, choose the agent in $\{a_1,a_2\}$ whoever values $B\setminus \{g_1\}$ plus part of $g_1$ by $\frac{2}{3}$ with a smaller fraction of $g_1$.
    Assume that agent is $a_1$. We allocate $B\setminus \{g_1\} \cup \{len_1(g_1,\frac{2}{3})\}$ to agent $a_1$ as $A_1$ and remove them from the instance.
    For the other three agents in $N'_r$, the value of $M'_r$ is at least $\frac{10}{3}$.
    As proved above, Condition (1) or (2) cannot be satisfied for the new reduced instance.

    \item \textit{Subsubcase 1.1.3:} Every agent in $N_r$ considers at least two of the items in $G_1$ divisible and there is no shared divisible item in $G_1$ for every two agents.
    So agent $a_4$ considers all the items in $G_1$ to be indivisible since $|G_1|=6$.
    We arbitrarily choose an item $g$ in $G_1$ to add to an empty bag $B$, and keep adding items in $G_2$ to $B$ until some agent values $B$ no less than $\frac{2}{3}$.
    Considering the proof of Subcase 1.1, Conditions (2), (3), and (6) must not simultaneously hold.

    \end{itemize}

    \textit{Subcase 2:} $|G_1|\geq 7$.
    Denote $G_1=\{g_1,g_2,\ldots,g_7,\ldots\}$.
    Assume the item in $A_4\cap G_1$ is $g_6$ and those in $\hat{G}_1$ are $\{g_1,\ldots,g_5\}$.
    Since item $g_7$ is not removed, only $a_4$ thinks $g_7$ is medium.
    Similar to the discussion of Subcase 1.1, we withdraw $A_4$ and $a_4$. $g_6$ is first put into the bag $B$, and the items in $G_2$ are added into $B$ until some agent considers the value of the bag to be no less than $\frac{2}{3}$.
    If there is one agent in $N_r$ who values $B$ by no less than $\frac{2}{3}$, allocate $B$ to her and remove them.
    Since Condition (4) shows any item in $\hat{G}_1$ is medium for any agent in $N_r$, there are 6 items $\{g_1,\ldots,g_6\}$ which are medium for at least one agent in $N'_r$, breaking Condition (3).
    Otherwise, the three agents consider the value of $B$ is less than $\frac{2}{3}$.
    By Claim \ref{cla:rem1}, Condition (1) or (2) is broken.

    \textit{Case 2:} $v_i(G_2) < \frac{2}{3} - v_i(g_0)$ for each $a_i \in N, g_0 \in G_1 ;|G_1| \geq 8$.
    In this case, if we find a temporary bundle, only one agent and two goods in $G_1$ are removed. 
    As mentioned in Lemma \ref{3.4}, the value of the temporary bundle for the remaining agents is less than $\frac{2}{3}$, which breaks Condition (1) or (2) by Claim \ref{cla:rem1}.

    Since two agents are removed in \textit{Case 3}, there is no need to discuss this case. 

    \textit{Case 4:} $v_i(G_2) < \frac{2}{3} - v_i(g_0)$ for each $a_i \in N, g_0 \in G_1; |G_1|=8$ and there are no shared divisible items in $G_1$.
    In this case, since no agent in $N_r$ values $A_4$ at more than $\frac{2}{3}$, Conditions (1) and (2) cannot be satisfied at the same time by Claim \ref{cla:rem1}.

    Based on the above discussion, Lemma \ref{lem:fix} is proved.
\end{proof}

    Combining everything together, the correctness of Theorem \ref{4agent} follows from Lemmata \ref{4-1} to \ref{lem:fix} stated above. 
    Therefore, Algorithm \ref{alg2} returns a $\frac{2}{3}$-MMS allocation for four agents in polynomial time.

\end{document}